\documentclass[10pt]{article}
\usepackage[utf8]{inputenc}
\usepackage{subcaption}
\usepackage{amsmath}
\usepackage{threeparttable}

\DeclareMathOperator*{\argmin}{arg\,min}
\usepackage{amssymb}
\usepackage{hyperref}
\usepackage{colortbl}
\usepackage[normalem]{ulem}
\usepackage{floatrow}
\usepackage[colorinlistoftodos]{todonotes}
\usepackage{graphicx}
\usepackage{tabularx,colortbl}
\usepackage{makecell}
\usepackage{multicol}
\definecolor{LightBlue}{RGB}{140,186,252}
\newcolumntype{a}{>{\columncolor{LightRed}}c}
\usepackage{amsthm,amsmath,amssymb,mathtools}
\usepackage{algorithm}
\usepackage{enumerate}
\usepackage{enumitem}
\usepackage{bbm}
\usepackage{titlesec}
\usepackage{algorithmic}
\usepackage{bigints}
\usepackage{setspace} %
\usepackage{verbatim} %
\usepackage{booktabs} %
\usepackage[width=0.9\textwidth]{caption}
\usepackage{fullpage}
\usepackage{hyperref}
\usepackage{natbib}
\usepackage[toc,page]{appendix}
\usepackage{indentfirst}

\def\1{\mathbbm{1}}

\newtheorem{theorem}{Theorem}
\newtheorem{proposition}{Proposition}
\newtheorem{corollary}{Corollary}

\newtheorem{lemma}{Lemma}
\newtheorem{remark}{Remark}

\usepackage{xcolor}

\title{Predicting data value before collection: A coefficient for prioritizing sources under random distribution shift}
\author{Ivy Zhang and Dominik Rothenh\"ausler}
\usepackage{fullpage}

\begin{document}

\maketitle

\begin{abstract}
Researchers often face choices between multiple data sources that differ in quality, cost, and representativeness. Which sources will most improve predictive performance? We study this data prioritization problem under a random distribution shift model, where candidate sources arise from random perturbations to a target population. We propose the Data Usefulness Coefficient (DUC), which predicts the reduction in prediction error from adding a dataset to training, using only covariate summary statistics and no outcome data. We prove that under random shifts, covariate differences between sources are informative about outcome prediction quality. Through theory and experiments on synthetic and real data, we demonstrate that DUC-based selection outperforms alternative strategies, allowing more efficient resource allocation across heterogeneous data sources. The method provides interpretable rankings between candidate datasets and works for any data modality, including ordinal, categorical, and continuous data.\end{abstract}

\section{Introduction}

Consider a practitioner aiming to train a predictive model for a target population but has limited resources for outcome data collection. They have access to covariate summary statistics from multiple candidate sources and must decide which sources to sample for the costly outcome data. This data collection prioritization problem is fundamental across domains: building a prediction model for California using data from other states, clinical trials selecting sites for patient recruitment, or companies deciding which customer segments to survey. Unlike active learning settings where practitioners decide whether to label individual units, we focus on batch source selection where the decision is whether to collect an entire dataset from a candidate source (e.g., setting up an experiment in a location).

One standard machine learning approach ranks candidate sources by distributional similarity of source to the target, using measures such as KL divergence, Wasserstein distance, or maximum mean discrepancy \citep{quinonero2008dataset, ben-david2006analysis, ben-david2010theory, kouw2018introduction, data_addition_dilemma}. However, these proxies have key limitations: they are often based on conservative upper bounds on generalization risk and may not correlate well with actual predictive improvements. Rather than relying on such bounds, we derive from rigorous theory a coefficient that directly predicts the reduction in excess risk from collecting outcome data from each candidate source, addressing the core question: which additional source will most improve predictive performance on the target population?

To illustrate our approach, suppose we are in 2017 and wish to predict cholesterol of U.S. adults. Outcome data collection is expensive but historical covariate data are readily available. As we demonstrate in Section \ref{sec: nhanes_pred}, a Kullback-Leibler divergence based approach and a domain classifier based approach would select D1 (dataset 1), a mixture of 2013 and 2015 observations, to use for our prediction task. However, D1 provides less predictive value for 2017 cholesterol prediction in reality. Our method, while not choosing the highest ranked dataset, selects D5 as seen in Figure \ref{fig:intro}, the next best option, and de-prioritizes weaker candidates using only covariate summary statistics by directly predicting risk reduction.

\begin{figure}[ht]
    \centering
    \includegraphics[width=0.4\linewidth]{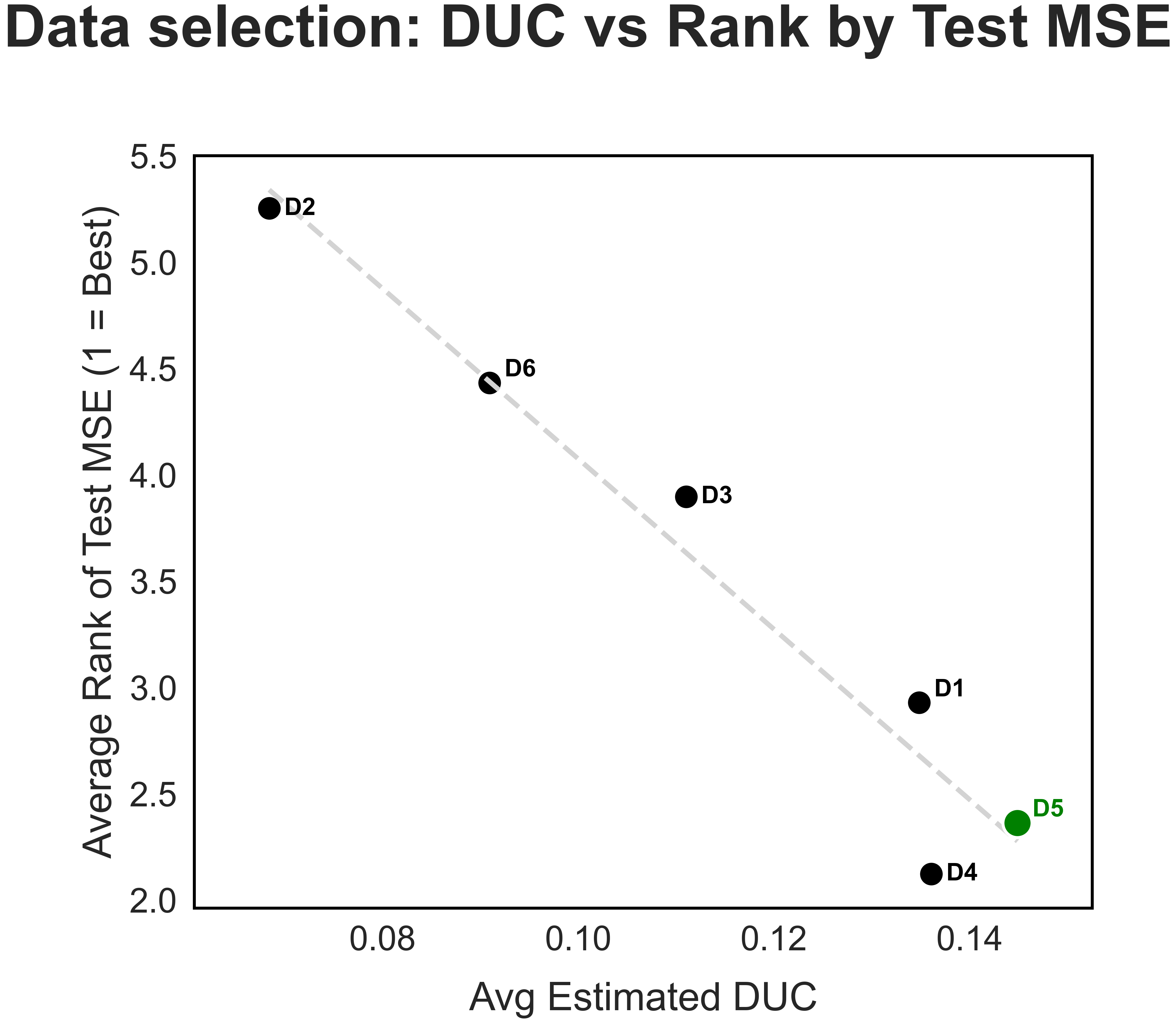}
        \caption{Relationship between our method, the Data Usefulness Coefficient (DUC; estimated using only covariate summary statistics) and the empirical rank (by test MSE) of a source dataset. Higher DUC corresponds to a better rank (1 = lowest test MSE). Results averaged over 1,000 trials. While our method missed D4 as the best ranking dataset, other competing methods select D1 as seen in Section \ref{sec: nhanes_pred}.}
    \label{fig:intro}
\end{figure}

\paragraph{Our contribution.} We introduce the \textit{Data Usefulness Coefficient} (DUC), which solves a fundamental challenge in data prioritization: predicting which datasets will improve predictive performance before collecting expensive outcome data. The DUC has several advantages that distinguish it from existing approaches:

\begin{itemize}
\item \textbf{Predicting risk reduction}: Unlike methods that yield conservative upper bounds, DUC directly predicts percentage excess risk reduction on an interpretable 0-1 scale.
\item \textbf{Minimal data requirements}: DUC requires only covariate summary statistics, not outcome data, enabling cost-effective data prioritization. Statistically, we leverage the fact that under random shift, covariate shift is predictive of shifts in the outcome conditional on covariates $( Y | X)$ \citep{random_shift3}.
\item \textbf{Statistical rigor}: We provide a consistent estimator with asymptotically valid confidence interval, quantifying both distributional and sampling variation.
\end{itemize}

\paragraph{Paper outline.} Section~\ref{sec: setting} introduces the random distribution shift framework that provides the foundation for our theoretical analysis. Related work is discussed in Section~\ref{sec:related-work}. Section~\ref{sec: partial_cor} presents our main theoretical results connecting data usefulness to partial correlations and derives the DUC. Section~\ref{sec: estimation} provides a consistent estimator for the DUC using only covariate summary statistics and establishes its asymptotic properties. Section~\ref{sec: empirical} validates our approach through simulations and real-world applications, including the 2017 cholesterol prediction task. Section~\ref{sec: discussion} discusses limitations and future directions. In addition, Appendix \ref{sec: optimal_samp} presents a method for optimal sampling under budget and size constraints, complementing the main results on data collection strategies within the random shift, empirical risk minimization framework.

\paragraph{Why we employ a random distributional shift model.}
To predict data value before collection, we need a framework that connects observable covariate information to unobserved prediction performance. Standard covariate shift models, where $P(Y|X)$ remains constant across sources, cannot address this challenge. If the conditional distribution is identical everywhere, covariate statistics reveal nothing about outcome prediction quality. Real-world shifts are more complex: for example, U.S. adult's cholesterol level distribution changes over time due to socioeconomic shifts, pharmacological advancements, and policy changes that simultaneously affect both who is in the population ($P(X)$) and the relationship between covariates and outcomes ($P(Y|X)$).
We therefore adopt a random distribution shift framework that models these multifaceted changes through random perturbations affecting both marginal and conditional distributions. This framework couples covariate and label shifts through shared distributional weights, enabling outcome-free prediction of data usefulness. We now formalize this framework.

\section{Setting}\label{sec: setting}

How should we model the relationship between candidate data sources and a target population? Many existing approaches assume covariate shift \citep{Shimodaira2000}, where only the marginal distribution $P(X)$ changes while $P(Y|X)$ remains fixed. However, this assumption is often violated in tabular datasets \citep{liu2023need}. Our approach builds on a random distribution shift model that captures complex, multifaceted distribution changes through simultaneous perturbations to both $P(X)$ and $P(Y|X)$. Distribution shift models often place structure on the likelihood ratios between target and source distributions such as in the covariate shift or domain-invariant representation learning literature \citep{Shimodaira2000, Ganin2015}. The random shift model described below instead assumes that the likelihood ratio between the target and source distributions are random.

To illustrate, first consider the U.S. adult cholesterol prediction example over time. Between 1999 and 2017, U.S. adult's cholesterol level distribution shifts due to numerous factors: lifestyle trends (e.g. smoking, diets, physical activity), demographic transitions (migration patterns), and policy changes (increased access to statins). Rather than stemming from a single systematic change, this shift arises from many small, independent perturbations across different population subgroups. Some demographic segments may see cholesterol level increases, and others decreases, in ways that are difficult to predict individually but collectively create complex distribution changes.

The random perturbation framework captures this phenomenon by modeling each source distribution $P_k$ as a random reweighting of different regions of an underlying population distribution $P_f$. Under this model, different parts of the data space are randomly assigned higher or lower probability mass compared to the base distribution, reflecting the intricate nature of real-world distribution shifts. 

We now formalize this intuition mathematically. Suppose we have access to some data from $P_1$ (target distribution) and $P_2,\ldots,P_{K}$ (source distributions) for training a model to predict an outcome $y^{(1)}$ based on covariates $X^{(1)}$, both drawn from the target population $P_1$, and covariates $X^{(2)}, \ldots, X^{(K)}$ drawn from the source distributions. We observe data $(X^{(k)}, y^{(k)}) \in \mathbb{R}^L \times \mathbb{R}$ from $n_k$ samples drawn i.i.d.\ from $P_k$, $k=1,\ldots,K$, and covariate-only data $X^{(K+1)}$ from a candidate source distribution $P_{K+1}$. The question at hand (which we will describe formally in the following section) is whether we should collect $n_{K+1}$ labeled observations $(X^{(K+1)}, y^{(K+1)}) \in \mathbb{R}^L \times \mathbb{R}$ from the candidate source distribution $P_{K+1}$. We are interested in answering this question based on only covariate data.

Each source distribution $P_k$ is assumed to be a random perturbation of a fixed underlying distribution $P_f$, as in the random distributional perturbation setting of \cite{random_shift1}.  The data space is partitioned into equal-probability regions $(I_j)_{j=1,\ldots,m}$, each assigned a random weight $W_j^{(k)}$ that determines how much probability mass that region receives under distribution $P_k$, that is
\begin{equation*}
    P_{k}(x,y) = \frac{W_j^{(k)}}{\frac{1}{m} \sum_{j'=1}^m W_{j'}^{(k)}} P_f(x,y) \qquad \text{ for } (x,y) \in  I_j.
\end{equation*}
The weights $W_j^{(\bullet)}$, $j=1,\ldots,m$ are i.i.d.\ positive random variables with mean one, capturing the idea that shifts are numerous, independent, and unbiased on average. 

\begin{figure}[ht]
    \centering
    \includegraphics[width=0.75\linewidth]{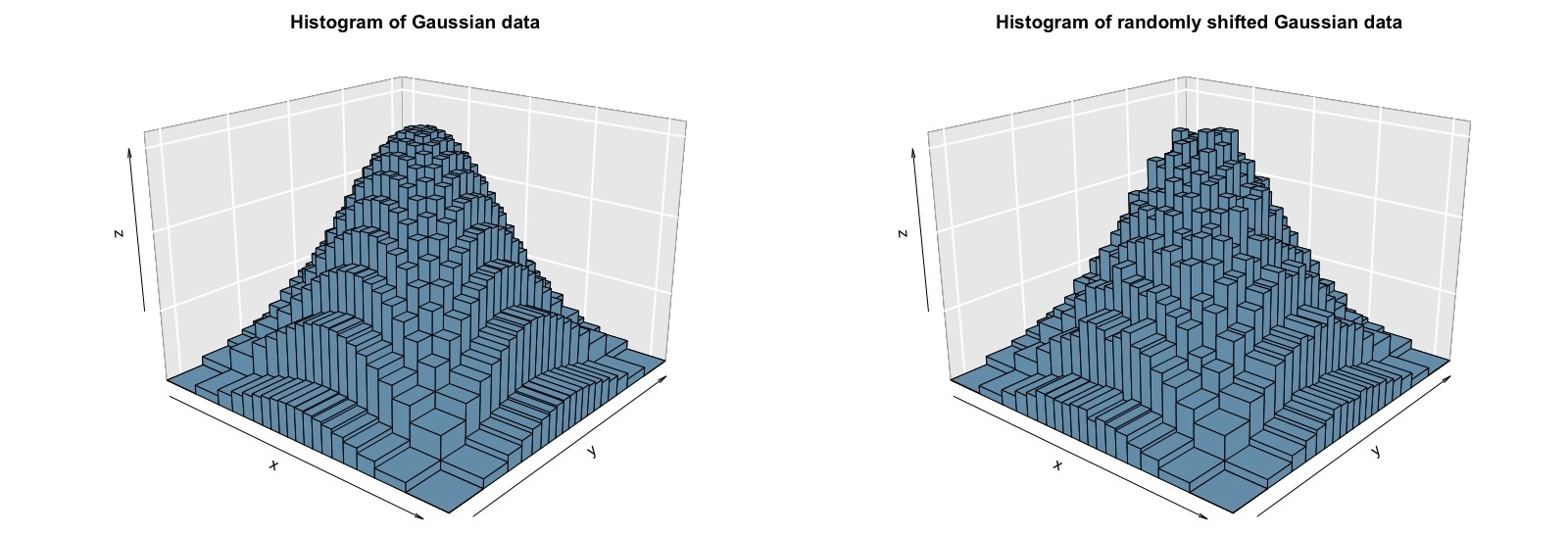}
        \caption{Visualization of the random shift model. The fixed distribution (left) is partitioned into regions $I_j$. The perturbed source distribution $P_k$ (right) is formed by reweighting each region with i.i.d random weights $W_j^{(k)}$. }
    \label{fig: random_shift}
\end{figure}

The asymptotic regime assumes that distributional uncertainty (governed by the number of regions $m$) and sampling uncertainty (governed by sample sizes $n_k$) are of the same order, ensuring that the trade-off between data quality and quantity is meaningful. Formally, we consider sequences $n_k(m)$ as $m \rightarrow \infty$ such that $n_k(m)/m \rightarrow c_k$, a constant that  measures the ratio between sampling and distributional uncertainty. In the following, for simplicity, we notationally drop the dependence of $n_k$ on $m$.

This framework allows both $P(X)$ and $P(Y|X)$ to change simultaneously across sources while maintaining tractable analysis. Empirical evidence supporting this modeling approach across diverse real-world datasets is provided by \cite{random_shift3}. Technical construction details, including regularity assumptions, are provided in Appendix \ref{app: random_shift_construction}.

\begin{remark}[Random shift enables outcome-free estimation of risk reduction]
The same distributional weights $W^{(k)}$ govern both covariate and label distributions. As we will discuss, this coupling implies that covariate shift predicts label shift, enabling us to estimate excess risk reduction from covariate data alone without observing outcomes from candidate sources. This property is unique to the random shift framework and does not hold under standard covariate shift assumptions.
\end{remark}

Our goal is to predict $y^{(1)}$ based on $X^{(1)}$ using data from the related but different distributions $P_k$, addressing the core question: which sources should we prioritize for outcome data collection before observing any outcomes? Having established the distributional framework, we now turn to the central methodological contributions of this work.

\paragraph{Notation.} 
Table~\ref{tab:notation} summarizes key notation used throughout the paper. Let $E_k$, $\mathrm{Var}_k$ be the expectation and variance under $P_k$ and $\hat{E}_{k}$ be a sample mean, e.g. for any function $f: \mathcal{D} \rightarrow \mathbbm{R}$, $\hat{E}_k[f(D^{(k)})] = \frac{1}{n_k}\sum_{i=1}^{n_k} f(D^{(k)}_i)$ where $D^{(k)}_i = (X_i^{(k)}, y_i^{(k)})$. Denote $E$ as the expectation with respect to both distributional and sampling noise.

\begin{table}[h]
\centering
\caption{Notation used throughout the paper.}
\label{tab:notation}
\begin{tabular}{@{}ll@{}}
\toprule
\textbf{Notation} & \textbf{Description} \\ \midrule
$P_1, P_k, P_{K+1}$ & Target, source $k$, and candidate distribution (all random) \\
$n_k$  & Sample size from distribution $P_k$ \\
$(X^{(k)}, y^{(k)})$ & Covariate-outcome pair from distribution $P_k$ \\
$m$ & Number of regions in data space partition \\
$W_\bullet^{(k)}$ & Distributional weight vector for source $k$ \\
$c_k = \lim m/n_k$ & Ratio of distributional to sampling uncertainty \\
$\beta^*, \beta'$ & Optimal weights for approximating $W^{(1)}, W^{(K+1)}$ \\
$\mathcal{E}^K$ & Excess risk using $K$ sources \\
$\rho^2_{K+1|1,\ldots,K}$ & Data usefulness coefficient (DUC) \\
$E_k[\cdot], \hat{E}_k[\cdot]$ & Expectation under $P_k$ and sample mean from source $k$ \\
\bottomrule
\end{tabular}
\end{table}

\section{Related work}\label{sec:related-work}

Our work connects to several research areas: the random distribution shift framework that underlies our theoretical results, distance measure-based methods for domain adaptation that represent our main competitors, data valuation methods that require outcomes, and active learning approaches that select individual observations rather than sources.

\paragraph{Random distribution shift.}
The random distribution shift framework was introduced by \citet{random_shift1}. \citet{random_shift3} show that random distribution shift implies a ``predictive role of covariate shift" and validate this finding on 680 studies across 65 sites. The random shift model was extended to out-of-distribution prediction via reweighting by \citet{random_shift2}. These prior works focus on combining data from fixed, available sources. Our work addresses a distinct, earlier-stage question: before collecting expensive outcome data, which sources should we prioritize for data collection? We develop the first outcome-free criterion for source prioritization under random distributional perturbations.

\paragraph{Divergence-based selection and domain adaptation.}
A common approach to source selection uses distributional distance measures. 
Domain adaptation theory provides generalization bounds relating target risk to distance measures between domains \citep{ben-david2006analysis,ben-david2010theory}, with extensions to multiple sources \citep{mansour2009domain, cortes2011sample}. Surveys synthesize these developments \citep{kouw2018introduction, wilson2020survey}.

Most closely related is \cite{data_addition_dilemma}, which proposes heuristic selection strategies based on distance measures including domain classifier scores. However, existing distance-based approaches have key limitations: they usually rely on upper bounds rather than expected performance, provide only ordinal rankings without interpretable magnitudes. Our DUC directly predicts percentage excess risk reduction on an interpretable 0-1 scale, accounts for both covariate and label shift through the random shift framework, and provides theoretical guarantees on expected performance rather than worst-case bounds.

\paragraph{Data valuation and selection with outcomes.}
Outcome-aware methods quantify the marginal predictive contribution of data points or datasets using influence functions \citep{koh2017understanding}, Shapley values \citep{ghorbani2019data, jia2019towards}, and robust variants \citep{kwon2021beta}. These approaches require both outcomes and a fitted model to assess data value. Empirically, adding datasets can hurt performance via spurious correlations \citep{compton2023more}. Our approach is fundamentally different: it is model-free (applicable to any smooth, strictly convex loss function and prediction task) and outcome-free (uses only covariate summary statistics). This enables ex ante prioritization of sources before the costly step of outcome collection, addressing a distinct stage of the data acquisition pipeline.

\paragraph{Active learning and optimal design.}
Active learning selects which observations to label to improve a predictor \citep{settles2009active, tong2001support, schohn2000less, gal2017deep, ash2019deep}, while active statistical inference optimize label efficiency for inference targets \citep{active_inf}. Classical optimal design \citep{fedorov1972theory, atkinson2007optimum, chaloner1995bayesian} and model-assisted survey sampling \citep{sarndal2003model, lohr2021sampling} optimize data collection. These approaches select from a pool of available unlabeled observations and typically assume covariate shift only, with fixed $P(Y|X)$. Our problem is complementary: we prioritize which sources to sample based on covariate data and allow for simultaneous shift in both $P(X)$ and $P(Y|X)$.

\section{Measuring data usefulness for empirical risk minimization}\label{sec: partial_cor}

With the random shift framework in place, we now address the core challenge of quantifying data usefulness. We introduce the ``\emph{data usefulness coefficient}" (DUC) to quantify the value of an additional dataset. We will show that the covariance structure of the distributional weights determines the reduction in excess risk, independent of the specific loss function or prediction task. In Section \ref{sec: estimation}, we show how the DUC can be estimated from covariate data alone without requiring any outcome information. 

To measure how much new data helps, we first establish how we will combine data from multiple sources. We use weighted empirical risk minimization, which finds optimal convex combinations of datasets for prediction under distribution shift.

Our goal is to estimate the target parameter $\theta_1 = \argmin_\theta E_1[\ell(\theta,X^{(1)},y^{(1)})]$, where the expectation is over the (random) target distribution $P_1$. Note that since $P_1$ is random, $\theta_1$ is random, too. Given data from $K$ sources, we consider a weighted estimator that combines information across sources:
\begin{equation*}
\hat \theta^K(\beta) = \arg \min_\theta \sum_{k=1}^K \beta_k \frac{1}{n_k} \sum_{i=1}^{n_k} \ell(\theta,X_i^{(k)}, y_i^{(k)}),
\end{equation*}
where $\beta_k \geq 0$ are weights summing to one that determine how much to rely on each source. The weights $\beta^* = \arg \min_{\beta^\intercal 1 = 1} \mathbb{E}[(W^{(1)} - \sum_k \beta_{k=1}^K W^{(k)})^2 ] + \sum_k c_k \beta_k^2$ minimize the asymptotic risk and can be estimated using cross-validation or explicit formulas from \citet{random_shift2}.

To quantify performance, we measure the excess risk of our estimator relative to the optimal target parameter. Under the random shift framework, this requires accounting for both sampling uncertainty (from finite samples) and distributional uncertainty (from random perturbations). We define the clipped excess risk as:
\begin{equation*}
    \mathcal{E}^{K} \coloneqq E[\text{CLIP}_{G/m}( R_1(\hat{\theta}^K(\beta^*))-R_1(\theta_1))],
\end{equation*}
where $R_1(\theta) = E_1[\ell(\theta, X^{(1)},y^{(1)})]$ is the (random) population risk on the target, and the function $\text{CLIP}_{G/m}(x) \allowbreak = \max(\min(x,G/m),-G/m)$ bounds the excess risk. 

Now, consider incorporating an additional $(K+1)$-th dataset into training. Our main theoretical result shows that the reduction in excess risk depends only on the covariance structure of the random weights and sample sizes, not on the specific choice of loss function (e.g., squared error, Gamma, Poisson, other smooth strictly convex losses). The proof can be found in the Appendix.

\begin{theorem}[Data usefulness coefficient predicts improvement in excess risk]\label{thm: pred_rho}
Under standard regularity conditions on the loss function (including smoothness, strict convexity; detailed in Appendix~\ref{app: thm1_conditions}) and the random shift model, the relative improvement in excess risk from adding the $(K+1)$-th source is:
\begin{equation*}
 \lim_{G \rightarrow \infty} \lim_{m,n_k \rightarrow \infty}  \frac{\mathcal{E}^{K} - \mathcal{E}^{K+1}}{\mathcal{E}^{K}}  = \rho_{K+1|1,2,\ldots,K}^2,
\end{equation*}
where the \textit{data usefulness coefficient} is defined as:
\begin{equation}\label{eqn: rho_defn}
    \rho^2_{K+1|1, \ldots,K} = \frac{(\mathrm{Cov}( W^{(1)} -  W^{\beta^*}, W^{(K+1)} - W^{\beta'}) +  \langle \beta^*, \beta'\rangle_C )^2 }{ (\mathrm{Var}(W^{(1)} - W^{\beta^*}) + \|\beta^*\|^2_C )  ( \mathrm{Var}(W^{(K+1)} - W^{\beta'}) + \|\beta'\|^2_C + c_{K+1} ) }.
\end{equation}
Here, $\beta^* = \arg \min_{\beta^\intercal 1 = 1} \mathbb{E}[(W^{(1)} - W^{\beta})^2 ] + \|\beta\|_{C}^{2}$, $\beta'= \arg \min_{\beta^\intercal 1 = 1} \mathbb{E}[(W^{(K+1)} - W^{\beta})^2 ] + \|\beta\|_{C}^{2} $, and $W^{\beta} = \sum_{k=1}^K \beta_k W^{(k)}$. $\langle \beta^*,\beta' \rangle_{C} = \sum_{k=1}^K \beta_k^*\beta_k' c_k$ and $\|\beta\|_{C}^{2} 
= \sum_{k=1}^{K} \beta_k^{2}c_k$ are the $c_k$-weighted inner products and squared norms respectively, and  $c_k = \lim \frac{m}{n_k}$ are the sample size to distribution shift ratios. Note that $\rho^2 \in [0,1]$ directly measures the fractional reduction in excess risk: higher values indicate more useful datasets.
\end{theorem}

\begin{remark}
     Although $c_k$ depends on $m$, we do not estimate $m$ itself. Instead, we estimate the joint quantity $\mathrm{Cov}(W^{(k)}, W^{(k')})/m$ directly from covariate summaries, as we will demonstrate in Section~\ref{sec: estimation} for DUC estimation.
\end{remark}
\begin{remark}
    The clipping introduced in the excess risk $\mathcal{E}^{K}$ ensures bounded functionals and well-defined expectations under weak convergence. The nested limits in Theorem~\ref{thm: pred_rho} reflect this: first $m, n_k \to\infty$ at fixed $G$, then $G\to\infty$.  
\end{remark}
Intuitively, the DUC formula captures how sample sizes and distributional structure jointly determine candidate source value:
\begin{itemize}
    \item $\mathrm{Cov}(W^{(1)}-W^{\beta^*}, W^{(K+1)}-W^{\beta'})$ captures the \textit{distributional co-variation} between target $P_1$ and candidate $P_{K+1}$ not already captured by existing sources $P_2, \ldots, P_K$.
    \item $\mathrm{Var}(W^{(1)} - W^{\beta^*})$ and $\mathrm{Var}(W^{(K+1)} - W^{\beta'})$ capture the \textit{residual distributional variation} after accounting for $P_2, \ldots, P_K$.
    \item $c_k$ for $k=1,\ldots, K$ are inversely proportional to current sample sizes: smaller existing samples (larger $c_k$) increase the marginal value of additional data from any source.
    \item $c_{K+1}$ is inversely proportional to the candidate sample size: larger $n_{K+1}$ (smaller $c_{K+1}$) increase the marginal value of additional data from candidate $P_{K+1}$.
\end{itemize}

These quantities interact to produce a closed-form measure of candidate value that depends only on distributional structure and sample sizes.
The theorem establishes that excess risk reduction is determined by the distributional weight covariances and is universal across loss functions and data types (continuous, ordinal, binary, categorical, mixed). This means the DUC can be estimated without observing outcomes, using only the distributional relationships between sources (see Section~\ref{sec: estimation}).
To illustrate how this general formula operates in concrete settings, we now examine a special case with independent distributional weights.

\begin{corollary}\label{corr: indepW} 
    Let $W^{(1)} \stackrel{\text{a.s.}}{=} 1$. Under independent distributional weights, $W^{(k)}$, the data usefulness coefficient from Equation~\eqref{eqn: rho_defn} simplifies to 

        \begin{equation*}
        \rho^2_{K+1|1, \ldots,K} =\frac{\frac{1}{\mathrm{Var}(W^{(K+1)})+c_{K+1}}}{\sum_{k=1}^{K+1} \frac{1}{\mathrm{Var}(W^{(k)})+c_{k}}}.
    \end{equation*}
The proof can be found in Appendix \ref{app: corr1}. 
\end{corollary} 

 To illustrate how $\rho^2$ can guide data collection, 
suppose we have the following scenarios, using the approximation $c_k \approx m/n_k$. Our goal is to predict $y^{(1)}$ based on $X^{(1)}$, both drawn from our target distribution $P_1$. We want to decide whether to sample from a new data source, $P_2$, for this task.

\paragraph{Example 1 (no shift).} 
Suppose there is no distributional uncertainty in the simplest case. We have $n_1$ samples from $P_1$ and can collect $n_2$ samples from $P_2$. Then the excess risk ratio after adding the $(K+1)$-th source is:
\begin{equation*}
    1- \text{DUC} = 1- \frac{n_2}{n_1 + n_2} =  \frac{n_1}{n_1 + n_2} = \frac{\frac{1}{n_1 + n_2}}{\frac{1}{n_1}}
\end{equation*}
As expected, in the absence of distributional shift, risk reduction is proportional to inverse sample size.

\paragraph{Example 2 (shift).} 
Now suppose we already have $n_2$ samples from a perturbed distribution $P_2$ with distributional variance $\mathrm{Var}(W^{(2)})/m$, and $n_1$ samples from our target $P_1$ with no distributional uncertainty. We can collect $100$ additional observations from $P_1$ (fixed target) or $200$ from $P_2$ (perturbed source): how much precision would we gain by sampling these additional observations? In this case, if we sample from $P_1$, the relative improvement is
\begin{equation*}
    \text{DUC} = \frac{100}{n_1 + 100 + \frac{1}{\frac{1}{n_2} + \mathrm{Var}(W^{(2)})/m } }
\end{equation*}
and if we sample from $P_2$ the relative improvement in MSE is
\begin{equation*}
    \text{DUC} = \frac{\frac{1}{\frac{1}{n_2+200} + \mathrm{Var}(W^{(2)})/m }- \frac{1}{\frac{1}{n_2} + \mathrm{Var}(W^{(2)})/m }}{n_1 + \frac{1}{\frac{1}{n_2+200} + \mathrm{Var}(W^{(2)})/m } }
\end{equation*}
Assume $\mathrm{Var}(W^{(2)})/m = 0.1$ for this example. \newline 
\textbf{Option A: 100 more target samples.} Going from $n_1 = 50$ to $n_1 = 150$, the relative improvement in MSE is:
\begin{align*}
\text{DUC} &= \frac{100}{150 + \frac{1}{\frac{1}{1000} + 0.1}}\approx 62.5 \%
\end{align*}
\newline
\textbf{Option B: 200 more perturbed samples.} Going from $n_2 = 1000$ to $n_2 = 1200$, the relative improvement in MSE is:
\begin{align*}
\text{DUC} &= \frac{\frac{1}{\frac{1}{1000} + 0.1} - \frac{1}{\frac{1}{1200} + 0.1}}{50 + \frac{1}{\frac{1}{1200} + 0.1}} \approx 0\%
\end{align*}
The DUC correctly predicts that target samples are far more valuable here. 

\paragraph{Example 3 (no sampling uncertainty)} On the opposite side of Example 1, suppose there is no sampling uncertainty, i.e. $n_k \rightarrow \infty, k>1$ and we have no samples from $P_1$, so distributional uncertainty dominates. This gives $\beta_{2:K}^* \approx \argmin_{\beta_{2:K}^\intercal 1 = 1} \mathbb{E}[( W^{(1)} - \sum_{k=2}^K \beta_k W^{(k)}  )^2]$ and $\beta_1^* \approx 0$. Define $W^{\beta_{2:K}^*} \coloneqq \sum_{k=2}^K \beta_k^* W^{(k)}$. Then we have
\begin{equation*}
    \rho^2 \approx \text{Cor}^2(W^{(1)} -  W^{\beta_{2:K}^*}, W^{(K+1)} - W^{\beta_{2:K}^*} |  (W^{(k)} - W^{\beta_{2:K}^*})_{k=2,\ldots,K} ).
\end{equation*}
i.e., we can interpret the value of the additional data as the partial correlation between the differences in the target distribution weight, $W^{(1)}$ and the candidate source distribution weight, $W^{(K+1)}$ with existing data from $P_2, \ldots, P_K$.

This framework allows practitioners to compare candidate datasets systematically, moving beyond ad-hoc similarity measures to principled risk-reduction predictions. The next section shows how to estimate $\rho^2$ in practice using only covariate summary statistics.

\section{Estimation of the data usefulness coefficient}\label{sec: estimation}
In practice, data prioritization decisions must often be made before collecting expensive outcomes. While outcome collection is costly, covariate summary statistics are typically available from candidate sources. 

Our estimation approach exploits this asymmetry: we estimate $\rho^2_{K+1|1,2, \ldots,K}$ by measuring how each candidate source's covariate distribution relates to the target population. Assuming target population summaries $E_1[X_\ell]$ are known (as in census or clinical applications), we construct vectors capturing distributional differences and compute partial correlations to isolate each source's unique contribution. Let $K < L$, i.e., there are more covariates than candidate source distributions. This assumption is natural in many applications (e.g., choosing among a few geographic regions using dozens of demographic variables) and ensures the covariance matrix has full rank for reliable partial correlation estimation. 

To connect Equation~\eqref{eqn: rho_defn} to practical estimation of the DUC and provide intuition, we derive an approximation that relates distributional weights to correlations between covariate means. Specifically, we can approximate the covariance between two distributional weights $W^{(k)}$ and $W^{(k')}$ in terms of the covariance and variance of covariates means:
\begin{align*}
    \mathrm{Cov}(E_k[X_\ell],E_{k'}[X_\ell]) &\approx \mathrm{Cov}\left(\sum_{j=1}^m \frac{W_j^{(k)}}{m} E_f[X_\ell|(X,Y) \in I_j],\sum_{j=1}^m \frac{W_j^{(k')}}{m} E_f[X_\ell|(X,Y) \in I_j]\right)\\
    &= \frac{1}{m^2} \sum_{j=1}^m \mathrm{Cov}(W_j^{(k)}, W_j^{(k')}) E_f[X_\ell|(X,Y) \in I_j]^2\\
    &\approx \frac{1}{m} \mathrm{Cov}(W^{(k)}, W^{(k')}) Var_f(X_\ell).
\end{align*}
Here, in the first line we used the definition of the random shift and that $\frac{1}{m} \sum_{j=1}^m W_j^{(k)} = 1 + o_P(1/\sqrt{m})$. In the second line we used that the weights are i.i.d.\ across $j=1,\ldots,m$. In the limit as the partition refines, $\frac{1}{m}\sum_{j=1}^m E_f[X_\ell|I_j]^2 \to \mathrm{Var}_f(X_\ell)$ when $X_\ell$ is centered under $P_f$.

This approximation reveals that the covariance structure of distributional weights can be estimated through the empirical variance of covariate means. The following proposition formalizes this connection and establishes consistency of the estimator (formal proof in Appendix \ref{app: consistency_cov}).

\begin{proposition}[Consistency and asymptotic normality of data usefulness estimate] \label{prop: consistency_cov}
Let $K < L$ and assume the setting of Lemma~\ref{lemma:samp_asymp} with the covariates standardized to have unit variance and be uncorrelated under $P_f$. Define the vectors 
$$Z^{(1)} = (E_1[X_\ell] - \hat{E}_1[X_\ell])_{\ell=1}^L, \quad Z^{(k)} = (\hat{E}_k[X_\ell] - \hat{E}_1[X_\ell])_{\ell=1}^L \text{ for } k = 2, \ldots, K+1.$$ 
Let $\hat{\rho}^2_{K+1|1,2, \ldots,K}$ denote the squared empirical partial correlation of $Z^{(1)}$ and $Z^{(K+1)}$ adjusted for $Z^{(2:K)}$, and $\hat{\sigma}^2 = 4\hat{\rho}^2_{K+1|1,2, \ldots,K}(1 - \hat{\rho}^2_{K+1|1,2, \ldots,K})^2$. Then,
\begin{equation*}
  \lim_{L \rightarrow \infty}  \lim_{m \rightarrow \infty} 
 P \left( \sqrt{L}\frac{\hat{\rho}^2_{K+1|1,2, \ldots,K} - \rho^2_{K+1|1,2, \ldots,K}}{\hat{\sigma}^2} \le x \right) = \Phi(x),
\end{equation*}
where $\Phi(x)$ is the cumulative density function of a standard Gaussian random variable.
\end{proposition}
This result provides a practical application of Theorem~\ref{thm: pred_rho}. While the theorem shows that datasets with higher $\rho^2_{K+1|1,2, \ldots,K}$ provide greater excess risk reduction, Proposition~\ref{prop: consistency_cov} demonstrates how to reliably estimate these values using only covariate summary statistics. The precision depends on $L$, the number of available covariates. More covariate information yields more precise DUC estimates, allowing for more confident data prioritization decisions.

\begin{remark}[Practical implementation with correlated covariates]\label{rem:whitening}
The assumption that covariates are uncorrelated under $P_f$ may not hold in practice. We recommend applying a whitening transformation to the covariates before computing DUC estimates. Specifically, estimate the covariance matrix $\hat{\Sigma}$ from available data (pooling across sources), compute its Cholesky decomposition $\hat{\Sigma} = L L^\intercal$, and transform covariates as $\tilde{X} = L^{-1}(X - \mu)$ where $\mu$ is the pooled mean. The DUC is then computed using the whitened covariates $\tilde{X}_\ell$, which by construction have unit variance and are approximately uncorrelated, satisfying the conditions of Proposition~\ref{prop: consistency_cov}.
\end{remark}

The proof of the following result can be found in Appendix \ref{app: ci_transf}.

\begin{remark}\label{confint}
When DUC estimates are similar, accounting for uncertainty is essential for reliable decision-making. Let $m, n_k \rightarrow \infty$ with $m/n_k \rightarrow c_k > 0$, assume non-zero sample variances for all $Z^{(k)}$ when $L > 2$ and $\hat{\rho} \in (0, 1)$, and define $g(\rho^2) = \frac{1}{2}\log \left(\frac{1+\rho}{1-\rho}\right)$ and $z_{1-\frac{\alpha}{2}}$ as the $1-\frac{\alpha}{2}$ quantile of the standard normal distribution. Then
\begin{equation}\label{eqn: CI}
    \left[g^{-1}\left(g(\hat{\rho}^2_{K+1|1, \ldots,K}) \pm L^{-1/2}z_{1-\frac{\alpha}{2}}\right)\right]
\end{equation} 
is an asymptotically valid confidence interval for $\rho^2_{K+1|1, \ldots,K}$ that accounts for both sampling and distributional variation. 
\end{remark}

This estimation framework provides a complete practical methodology for data prioritization. Given target population summary statistics and covariate information from candidate sources, practitioners can compute DUC estimates, assess their uncertainty through confidence intervals, and make informed data collection decisions. The approach scales with the number of available covariates, making it particularly valuable in high-dimensional settings.

\section{Empirical examples} \label{sec: empirical}

Having developed the theoretical foundation and estimation procedures for the DUC, we now validate our approach through empirical studies designed to answer two key questions:

\begin{enumerate}
    \item \textbf{Does DUC accurately predict data value?} We compare DUC-based rankings against baseline methods (KL divergence, domain classifier scores) and actual predictive performance on real data (Section~\ref{sec: acs_income} and Section~\ref{sec: nhanes_pred}).
    \item \textbf{Do theoretical predictions hold in simulations?} We validate the core theoretical results through controlled simulations where ground truth is known (Section~\ref{sec:simulation}).
\end{enumerate}

Our empirical strategy progresses from real-world applications that demonstrate practical utility to controlled simulations that verify theoretical predictions. All experiments use only covariate information for computing DUC, with outcomes reserved solely for validation\footnote{Code to reproduce experiments can be found on \url{https://github.com/yzhangi96/prioritizing_data_collection}.}.

\subsection{ACS Income data}\label{sec: acs_income}
The ACS Income data from the American Community Survey (ACS) provides detailed demographic and socioeconomic information including individual income, employment status, education, and state. We demonstrate our data prioritization methods using a Random Forest (RF) model across various prediction tasks. For each setting, we use the preprocessed ACS data from the Python \texttt{folktables} package \citep{folktables}, reporting results on a held-out target test set ($n_\text{test} =1000$) over 1,000 trials. The random shift assumption for this dataset has been validated using diagnostic tests in \cite{random_shift2}, supporting our modeling framework. 

We begin with predicting the (log) income of adults in CA using data from other states. Suppose we have survey data for only 30 adults from CA, as surveying in CA is expensive. However, we can readily survey individuals from other states, such as NY, NJ, TX, IL, NC, FL, GA, MI, OH, and PA. Which state should we choose?

\begin{figure}[h]
   \centering
   \includegraphics[width=0.85\linewidth]{fig3.png}
   \caption{Average ranking, from lowest to highest test MSE, using Random Forest trained on California ($n_\text{CA} =30$) plus one source state ($n_\text{source} =1000$) over 1000 trials. Higher $\hat{\rho}^2$ (DUC) effectively corresponds to better ranking, where 1 represents the dataset with the lowest test MSE. Our method is most correlated with average rankings and correctly prioritizes datasets for prediction.}
      \label{fig: pred_CA}
\end{figure}

We compare the DUC against two baseline methods for ranking data sources: (1) KL divergence between source and target covariate distributions and (2) the domain classifier approach from \cite{data_addition_dilemma}, which estimates distributional distance by training a classifier to distinguish between source and target samples using covariates. Specifically, this approach estimates $\mathbb{E}_{\text{source}}[s(x)]$, where $s(x)$ is a classifier trained on covariates alone to predict whether samples come from the source distribution or target distribution. Higher scores indicate greater distributional dissimilarity.

Based on our DUC, we would survey NY or NJ, rather than PA or OH (Figure~\ref{fig: pred_CA}), which effectively results in better average predictive performance. The results are in line with \cite{data_addition_dilemma}, which demonstrate that the domain classifier score is more correlated with predictive performance compared to the KL divergence. However, while both competing methods would also select NY or NJ, it would also incorrectly overestimate the usefulness of states like TX. As a result, competing methods show weaker correlation with dataset rankings. 

We also highlight the importance of subsampling in practice for estimating DUC: we recommend employing subsampling and averaging results over repeated trials as we have shown in Figure \ref{fig: pred_CA} for estimation. We show in Appendix \ref{app: acs} results for a single trial (no subsampling) and 100 trials (mild subsampling) as comparison, which result in weaker correlation between DUC and ranking. In addition, we illustrate how subsampling can improve estimation even when we have poor estimates of target population covariate means.

\subsection{NHANES}\label{sec: nhanes_pred}
The CDC National Health and Nutrition Examination Survey (NHANES) dataset combines questionnaires with lab tests from U.S. participants annually. We demonstrate a second application of our data prioritization method on the task of predicting cholesterol level in adults in the year 2017 using subsets of data from prior years. As we will see, this setting which involves both temporal and demographic shifts is a more challenging one than the one for California income prediction. For this experiment, we use the preprocessed NHANES cholesterol dataset from the Python \texttt{tableshift} package \citep{gardner2023tableshift} and follow the same procedures as in the ACS Income setup for 1000 trials.

\begin{figure}[h]
    \centering
    \includegraphics[width=0.85\linewidth]{fig_nhanes.png}
    \caption{Average ranking (by test MSE) for training on 2017 (\(n_{2017}=30\)) and a source dataset (\(n_{\text{source}}=1000\)) across 1000 trials. Although none of the methods identified D4 as the top option, DUC exhibited the highest correlation with the empirical average ranks.}
        \label{fig: fig_nhanes}
\end{figure}

Suppose we are in 2017 and already collected 30 observations from adults. We can readily purchase historical data of various compositions. For example, one dataset contains questionnaires and lab results from a mix of records from 1999, 2001, 2003, and 2015 while another dataset contains only records from men in 2011. Which should we choose? In this more challenging task, none of the methods precisely identified dataset 4 (D4) as the best ranking dataset. However, our method still correctly prioritizes datasets such as datasets 4, 5, and 1 over datasets 3, 4, and 2 unlike the KL divergence or domain classifier score. The results are in line with both the ACS Income data from Section \ref{sec: acs_income} and \cite{data_addition_dilemma}, which demonstrate that the domain classifier score is more correlated with predictive performance compared to the KL divergence. See Appendix \ref{app: nhanes} for additional empirical results (e.g. test MSE plots) and description of candidate source data sets.

\subsection{Simulations}\label{sec:simulation}
We also demonstrate our main theoretical results with simulations. Datasets are generated from randomly shifted distributions following the model described in Section \ref{sec: setting}, using the \texttt{calinf} R package. Specifically, let covariates $X_1^{(k)}, \ldots, X_{15}^{(k)}$ be drawn from randomly shifted $\text{Bern}(0.5)$ distributions, and $X_{16}^{(k)}, \ldots, X_{30}^{(k)}$ from randomly shifted $N(0,1)$ distributions. The response is defined as $y^{(k)}=\sum_{j=1}^{15} X_j^{(k)} + (X_{16}^{(k)})^2+(X_{17}^{(k)})^2+ \sum_{j={18}}^{30} X_j^{(k)} + \epsilon_{\text{noise}}$, $\epsilon_{\text{noise}}$ drawn from randomly shifted $\text{Unif}(-1,1)$. Our goal is to predict $y^{(1)}$, the outcome from the target distribution $P_1$, with our covariates, with ordinary least squares regression.

Suppose we already have $300$ samples from a fixed $P_1$ and $400$ samples from a perturbed $P_2$. Using only covariate means to estimate the DUC of adding samples from $15$ different perturbed distributions of varying shifts and sample sizes, we compare the DUC with the actual model performance of optimally weighted samples. Figure \ref{fig: sim} shows the results of the simulation assessed on 50,000 unseen test samples from $P_1$, averaged over 1,000 trials: the estimated DUC is approximately that of the $1-\text{relative excess MSE ratio}$, per Theorem \ref{thm: pred_rho}. We also compare this with other proposed metrics such as KL divergence and a domain classifier score using covariates. Unlike the DUC, they do not correlate as well with actual excess relative MSE and can't be used to as confidently predict model performance. See the accompanying code for details. 

\begin{figure}[ht]
    \centering
    \includegraphics[width=\linewidth]{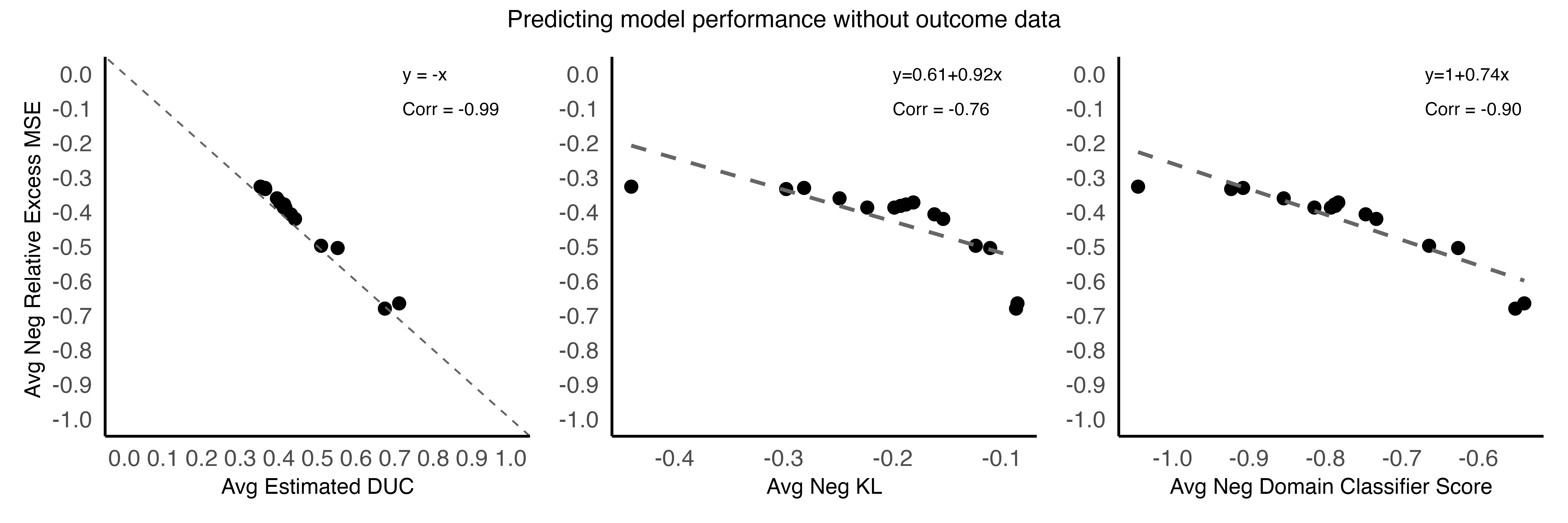}
    \caption{Simulation across 15 candidate datasets of varying sizes drawn from randomly perturbed distributions. With an existing training pool of \(300\) target samples from \(P_1\) and \(400\) perturbed source samples from \(P_2\), the DUC accurately predicts the average improvement from adding a new dataset in a nonlinear setting with mixed (continuous/binary) features. Compared with the KL divergence and domain classifier score, DUC aligns most closely with model test performance.}
    \label{fig: sim}
\end{figure}

\section{Discussion} \label{sec: discussion}

This paper introduced the data usefulness coefficient, $\rho^2$, which quantifies the reduction in excess risk from an additional data source under a random shift model. We established theoretical guarantees for $\rho^2$ estimation from limited source data and validated its effectiveness for data prioritization on tabular prediction tasks.

Several limitations of our approach warrant consideration. In the high-dimensional regime where the number of candidate datasets $K$ exceeds the covariate dimension $L$, regularization techniques may be required for stable covariance matrix estimation. Beyond the ERM setting, accurately predicting absolute MSE changes becomes more difficult; for instance, with tree-based models like Random Forests that do not directly minimize a loss function. However, our empirical results show that $\rho^2$ maintains accurate relative rankings of dataset usefulness even when these assumptions are violated.

Our framework applies to distribution shifts driven by random perturbations in the data generating process, shifts where source and target distributions differ due to accumulated variations in latent factors rather than systematic covariate shifts or adversarial perturbations. Future work could develop hybrid models that incorporate both random shift components and structured shift types (e.g., covariate shift, label shift, adversarial shift) to handle settings where distribution changes arise from multiple mechanisms.

\section{Acknowledgments}
We thank Tim Morrison, Asher Spector, Kevin Guo, Disha Ghandwani, Aditya Ghosh,  Markus Kalisch, and Nicolas Gubser for helpful feedback. Rothenh\"ausler gratefully acknowledges support as a David Huntington Faculty Scholar, Chamber Fellow, and from the Dieter Schwarz Foundation. 

\bibliographystyle{apalike}
\bibliography{ref}

\section*{Appendix}
\appendix
\section{Details of main results}

\subsection{Random shift construction details}\label{app: random_shift_construction}

This section provides the technical construction of the random distribution shift model referenced in Section \ref{sec: setting}. Without loss of generality, the distributional perturbations are constructed for uniform distributions on $[0,1]$. A result from probability theory shows that any random variable $D = (X,y)$ on a standard Borel space $\mathcal{D}$ can be written as a measurable function $D \stackrel{d}{=} h(U)$, where $U$ is a uniform random variable on $[0,1]$. The asymptotic behavior will not depend on the choice of this measurable function. With the transformation $h$, the construction of distributional perturbations for uniform variables can be generalized to the general cases by setting
\begin{equation*}
        P_k( (X,y) \in \bullet) = P_k(h(U) \in \bullet).
\end{equation*}
To construct a perturbed distribution $P_k$ for $U$, a uniform random variable under $P_f$, we first divide the $[0,1]$ range into $m$ bins: $I_j = [(j-1)/m, j/m]$ for $j = 1, \dots, m$. Let $W_{1}^{(k)}, \dots, W_{m}^{(k)}$ be i.i.d.\ positive random variables with mean one and finite variance. We assume the weights are bounded away from $0$, that is, we assume that $W_j^{(k)}\geq w_f$ for any $k$ and some fixed $w_f>0$. For each bin $I_j$, we assign a weight $W_j^{(k)}$ to represent how much probability mass is assigned to it in $P_k$. 

Thus, we define the randomly perturbed distribution $P_k$ by setting the density of the perturbed distribution,
\begin{equation*}
    P_{k}(x,y) = \frac{W_j^{(k)}}{\frac{1}{m} \sum_{j'=1}^m W_{j'}^{(k)}} P_f(x,y) \qquad \text{ for } (x,y) \in  I_j.
\end{equation*} 
As a result, some regions in the data space of $P_k$ have randomly higher or lower probability mass compared to $P_f$.   

Conditionally on $W = \left( W_{j}^{(k)} \right)_{k=1,\ldots,K; j=1,\ldots,m}$, $\{ D_{1}^{(k)}, \ldots, D_{n_k}^{(k)} \}$, $D^{(k)} \in \mathcal{D}$, are i.i.d.\ $n_k$ draws from $P_k$. Let $\frac{m}{n_k}$ converge to some constant $c_k > 0$ for each $k$. Formally, we consider sequences $n_k(m)$ as $m \rightarrow \infty$ such that $n_k(m)/m \rightarrow c_k > 0$, a constant that measures the ratio between sampling and distributional uncertainty. Intuitively, this assumption means that sampling uncertainty (order $1/n_k$) and distributional uncertainty (order $1/m$) are of the same order. In the following, for simplicity, we notationally drop the dependence of $n_k$ on $m$.  We chose this asymptotic regime because, in both other regimes (where either sampling uncertainty or distributional uncertainty dominates), the trade-off between data quality and data quantity is not apparent.  See \cite{random_shift1} and \cite{random_shift3} for diagnostics to test the random shift assumption in real data.

\subsection{Technical conditions for Theorem~\ref{thm: pred_rho}}\label{app: thm1_conditions}

The following regularity conditions are required for Theorem~\ref{thm: pred_rho}:

\begin{enumerate}
\item Let $\Theta$ be a compact set and $\ell(\theta,X,y)$ be any twice continuously differentiable, bounded loss function.
\item The loss function satisfies the Lipschitz condition: $\|\nabla^2\ell(\theta,X,y)-\nabla^2\ell(\theta',X,y)\|_{\text{op}} \leq L(X,y)\|\theta-\theta'\|$ for all $\theta, \theta'$ near $\theta_f$, the unique minimizer of the population risk under $P_f$, and some $L \in L^2(P_f)$.
\item The Hessian $E_f[\nabla^2\ell(\theta_f,X^{(f)},y^{(f)})]$ exists and is invertible.
\item The third order partial derivatives of $\ell(\theta, X,y)$ are dominated by a fixed function $h(\cdot) \in L^1(P_f)$ in an $\epsilon$-ball around $\theta_f$.
\item The asymptotic regime assumes $m/n_k \rightarrow c_k$ for some constants $c_k > 0$, $k=1,\ldots,K+1$.
\item The optimal weight vectors $\beta^*$ and $\beta'$ as defined in Theorem~\ref{thm: pred_rho}, satisfy $\beta_k^*>0$ and $\beta_k'>0$ for all $k=1,\ldots,K$.
\end{enumerate}

These conditions ensure that the empirical risk minimization estimators have well-defined asymptotic distributions and that the excess risk decomposition leading to the DUC formula is valid.

\subsection{Proof of Theorem~\ref{thm: pred_rho}}\label{app: thm1}

\begin{proof} The first part of this proof up until Eqn~\eqref{app: thm1_thetahat} is analogous to the proof of Lemma 3 from \cite{random_shift2}, which considers ERM under the regime where sampling uncertainty is of a lower order than distributional uncertainty, under a fixed target distribution. In contrast, in this paper, sampling uncertainty and distributional uncertainty are of the same order, and we consider a more general, random target distribution setting.

Consider the following weighted estimator 
    \begin{equation*}
        \hat \theta^K(\beta) = \arg \min_\theta \sum_{k=1}^K \beta_k \frac{1}{n_k} \sum_{i=1}^{n_k} \ell(\theta,X_i^{(k)}, y_i^{(k)}).
    \end{equation*} 
Let $\hat{E}_k[\ell(\theta,X^{(k)},y^{(k)})] = \frac{1}{n_k} \sum_{i=1}^{n_k} \ell(\theta,X_i^{(k)}, y_i^{(k)})$. By Lemma~\ref{lemma: mestimator_consistency}, $\hat{\theta}$ is a consistent estimator of $\theta_f = \argmin_\theta E_f[\ell(\theta, X^{(f)},y^{(f)})]$.
For a random $\tilde{\theta}$ between $\theta_f$ and $\hat{\theta}$, Taylor expansion around $\theta_f$ gives 
\begin{align}\label{eqn: thetahat_taylor}
\begin{split}
 0=\sum_{k=1}^K \beta_k \hat{E}_{k}[\nabla \ell(\hat{\theta},X^{(k)}, y^{(k)})] &= \sum_{k=1}^K \beta_k \hat{E}_{k}[\nabla \ell(\theta_f,X^{(k)}, y^{(k)})] + \bigg[\sum_{k=1}^K \beta_k \hat{E}_{k}[\nabla^2 \ell(\theta_f,X^{(k)}, y^{(k)})] \bigg](\hat{\theta}-\theta_f) \\
 &+ \frac{1}{2}(\hat{\theta}-\theta_f)^\intercal \sum_{k=1}^K \beta_k \hat{E}_{k}[\nabla^3 \ell(\tilde{\theta},X^{(k)}, y^{(k)})](\hat{\theta}-\theta_f).
 \end{split}
\end{align} 
We have that $\hat{\theta} = \theta_f + O_p(1/\sqrt{n_k})$ by Lemma~\ref{lemma: mestimator_consistency}, $\sum_{k=1}^K \beta_k \hat{E}_{k}[\nabla^2 \ell(\theta_f,X^{(k)}, y^{(k)})] = E_f[\nabla^2 \ell(\theta_f,X^{(f)}, y^{(f)})] + O_p(1/\sqrt{m})$ by Lemma~\ref{lemma:samp_asymp}, and $\|\sum_{k=1}^K\beta_k\hat{E}_{k}[\nabla^3 \ell(\tilde{\theta},X^{(k)}, y^{(k)})]\| \leq \sum_{k=1}^K|\beta_k|\frac{1}{n_k}\sum_{i=1}^{n_k}h(D_{i}^{(k)}) = O_p(1)$ by assumption. Since $m=O(n_k)$ for all $k$, rearranging Eqn~\eqref{eqn: thetahat_taylor} gives
\begin{align}\label{app: thm1_thetahat}
\begin{split}
 \hat{\theta}(\beta) &= \theta_f - \bigg[E_f[\nabla^2 \ell(\theta_f,X^{(f)}, y^{(f)})] + O_p(1/\sqrt{m})] \bigg]^{-1} \sum_{k=1}^K \beta_k \hat{E}_{k}[\nabla \ell(\theta_f,X^{(k)}, y^{(k)})] + O_p(1/m)\\
 &=\theta_f - \bigg[E_f[\nabla^2 \ell(\theta_f,X^{(f)}, y^{(f)})] \bigg]^{-1} \sum_{k=1}^K \beta_k \hat{E}_{k}[\nabla \ell(\theta_f,X^{(k)}, y^{(k)})] + O_p(1/m).
\end{split}
\end{align} 
By an analogous argument, we get
\begin{align}\label{app: thm1_theta0}
\begin{split}
 \theta_1 &= \theta_f - \bigg[E_f[\nabla^2 \ell(\theta_f,X^{(f)}, y^{(f)})] \bigg]^{-1}   E_{1}[\nabla \ell(\theta_f,X^{(1)}, y^{(1)})] + O_p(1/m).
\end{split}
\end{align} 

On the high probability event that $\hat{\theta}$ and $\theta_1$ are close enough to $\theta_f$ for the L-Lipschitz assumption to hold (see Corollary~\ref{cor: mestimator_consistency_theta0}), a Taylor expansion of the excess risk, $E_1[\ell(\hat{\theta}(\beta),X^{(1)}, y^{(1)}) - \ell(\theta_1,X^{(1)}, y^{(1)})]$, gives
\begin{align*}
  &E_1[\ell(\hat{\theta}(\beta),X^{(1)}, y^{(1)}) - \ell(\theta_1,X^{(1)}, y^{(1)})]\\ &= E_1[\nabla \ell(\theta_1,X^{(1)}, y^{(1)})] ( \hat{\theta}(\beta) - \theta_1) + \frac{1}{2} (\hat{\theta}(\beta) - \theta_1)^\intercal \nabla^2 \ell(\tilde{\theta},X^{(1)}, y^{(1)})(\hat{\theta}(\beta) - \theta_1)]\\
  &= \frac{1}{2} (\hat{\theta}(\beta) - \theta_1)^\intercal E_1[\nabla^2 \ell(\tilde{\theta},X^{(1)}, y^{(1)})](\hat{\theta}(\beta) - \theta_1)\\
  &=\frac{1}{2} (\hat{\theta}(\beta) - \theta_1)^\intercal E_1[\nabla^2 \ell(\theta_f,X^{(1)}, y^{(1)})](\hat{\theta}(\beta) - \theta_1)+r_m,
\end{align*}
where the remainder term, $r_m=\frac{1}{2} (\hat{\theta}(\beta) - \theta_1)^\intercal E_1[\nabla^2 \ell(\tilde{\theta},X^{(1)}, y^{(1)})- \nabla^2 \ell(\theta_f,X^{(1)}, y^{(1)})](\hat{\theta}(\beta) - \theta_1)$. We have by Corollary~\ref{cor: mestimator_consistency_theta0} that $\hat{\theta} = \theta_1 + O_p(1/\sqrt{m})$ and %
\begin{align*}
     E_1[\nabla^2 \ell(\tilde{\theta},X^{(1)}, y^{(1)})- \nabla^2 \ell(\theta_f,X^{(1)}, y^{(1)})] &\leq E_1[L(x)\|\tilde{\theta}- \theta_f\|] && \text{(Lipschitz assumption)}\\
     &\leq E_1[L(x)]\max(\|\tilde{\theta}- \theta_f\|, \|\hat{\theta}- \theta_f\|)\\
     &=O_p(1) \cdot O_p(1/\sqrt{m}),
\end{align*}
so we get that $r_m=o_p(1/m)$. Once more, we get
\begin{align*}
  E_1[\ell(\hat{\theta}(\beta),X^{(1)}, y^{(1)}) - \ell(\theta_1,X^{(1)}, y^{(1)})] 
  &=\frac{1}{2} (\hat{\theta}(\beta) - \theta_1)^\intercal E_1[\nabla^2 \ell(\theta_f,X^{(1)}, y^{(1)})](\hat{\theta}(\beta) - \theta_1)+o_p(1/m)\\
  &=\frac{1}{2} (\hat{\theta}(\beta) - \theta_1)^\intercal E_f[\nabla^2 \ell(\theta_f,X^{(f)}, y^{(f)})](\hat{\theta}(\beta) - \theta_1) + r'_m +o_p(1/m).
\end{align*}
Here, $r'_m=\frac{1}{2} (\hat{\theta}(\beta) - \theta_1)^\intercal (E_1[\nabla^2 \ell(\theta_f,X^{(1)}, y^{(1)})] - E_f[\nabla^2 \ell(\theta_f,X^{(f)}, y^{(f)})])(\hat{\theta}(\beta) - \theta_1)$. By Lemma~\ref{lemma:samp_asymp}, $E_1[\nabla^2 \ell(\theta_f,X^{(1)}, y^{(1)})] = E_f[\nabla^2 \ell(\theta_f,X^{(f)}, y^{(f)})] + O_p(1/\sqrt{m})$, so $r'_m=o_p(1/m)$ as well. Denote the vector $\hat{E}_{1:K}[\nabla \ell(\theta_f,X, y)]^\intercal \coloneqq [\hat{E}_{1}[\nabla \ell(\theta_f,X^{(1)}, y^{(1)}), \ldots, \hat{E}_{K}[\nabla \ell(\theta_f,X^{(K)}, y^{(K)})]$. Using Eqn~\eqref{app: thm1_thetahat} and Eqn~\eqref{app: thm1_theta0}, we therefore get
\begin{align*}
     &E_1[\ell(\hat{\theta}(\beta),X^{(1)}, y^{(1)}) - \ell(\theta_1,X^{(1)}, y^{(1)})]\\
  &=\frac{1}{2} (\hat{\theta}(\beta) - \theta_1)^\intercal E_f[\nabla^2 \ell(\theta_f,X^{(f)}, y^{(f)})](\hat{\theta}(\beta) - \theta_1) + o_p(1/m)\\
    &=\beta^\intercal \left( \hat{E}_{1:K}[\nabla \ell(\theta_f,X, y)]-E_1[\nabla \ell(\theta_f,X^{(1)}, y^{(1)})] \right)^\intercal  S^{-1}_f \left(\hat{E}_{1:K}[\nabla \ell(\theta_f,X, y)] -E_1[\nabla \ell(\theta_f,X^{(1)}, y^{(1)})]\right) \beta+o_p(1/m),
\end{align*}
where $S^{-1}_f\coloneqq \frac{1}{2}E_f[\nabla^2 \ell(\theta_f,X^{(f)}, y^{(f)})]^{-1}$. 
 By Lemma~\ref{lemma:samp_asymp}, we have that for $$J_{m,K}  \coloneqq \left(\hat{E}_{1:K}[\nabla \ell(\theta_f,X, y)]  -E_1[\nabla \ell(\theta_f,X^{(1)}, y^{(1)})] \right)^\intercal  S^{-1}_f \left(\hat{E}_{1:K}[\nabla \ell(\theta_f,X, y)]  -E_1[\nabla \ell(\theta_f,X^{(1)}, y^{(1)})]\right),$$ $mJ_{m,K} \overset{d}{\rightarrow} J_K$, for some random variable $J_K$ with mean $\Sigma^{K} \text{tr}(S^{-1}_f\mathrm{Var}_f(\nabla \ell(\theta_f,X^{(f)}, y^{(f)})))$ and finite variance. Here, $\Sigma^K= A\Sigma^{W}A^\intercal + \mathrm{Diag}(c_1, c_2,\ldots,c_K)$ is the limiting covariance matrix, with
\begin{equation*}
    A = \begin{pmatrix}
     0_{1 \times 1} & 0_{1 \times (K-1)} \\
-1_{(K-1) \times 1} & I_{(K-1) \times (K-1)}
    \end{pmatrix},
\end{equation*}
as in Lemma \ref{lemma:samp_asymp}. Then we have 
\begin{align}
\begin{split}\label{eq:weak-convergence-excess-risk}
     \lim_{G\rightarrow \infty} \lim_{m\rightarrow \infty} m\mathcal{E}^{K}
     &= \lim_{G\rightarrow \infty} \lim_{m\rightarrow \infty} mE[\max(E_1[\ell(\hat{\theta},X^{(1)}, y^{(1)}) - \ell(\theta_1,X^{(1)}, y^{(1)})], -G/m) \wedge G/m] \\
     &= \lim_{G\rightarrow \infty} \lim_{m\rightarrow \infty} mE[\max(\beta^\intercal J_{m,K} \beta + o_p(1/m), -G/m) \wedge G/m] \\
          &= \lim_{G\rightarrow \infty} E\bigg[\max \left(\beta^\intercal J_K\beta, -G \right) \wedge G \bigg] \\
     &= \beta^\intercal \bigg( \Sigma^{K}  \text{tr}(S^{-1}_f \mathrm{Var}_f(\nabla \ell(\theta_f,X^{(f)}, y^{(f)})))\bigg)\beta,
  \end{split}
  \end{align}
where we use dominated convergence theorem to interchange the limits and expectation and continuous mapping for $f(x)=\max(x, -G)$. 

Recall that we assume the optimal $\beta$'s are positive, so we remove the $\beta \geq 0$ constraint. Solving the optimization problem with a $\beta^\intercal \mathbf{1}=1$ constraint gives optimal weights in the limit of 
\begin{equation}\label{eqn: optimal_beta}
    \beta^* = \frac{(\Sigma^K)^{-1}\mathbf{1}}{\mathbf{1}^\intercal (\Sigma^K)^{-1}\mathbf{1}}    
\end{equation}
Substituting the optimal weights of Eqn \ref{eqn: optimal_beta} into the objective of Eqn \ref{eq:weak-convergence-excess-risk} gives
\begin{align*}
    \lim_{G\rightarrow \infty} \lim_{m\rightarrow \infty} m\mathcal{E}^{K} &=\beta^\intercal \bigg( \Sigma^{K}  \text{tr}(S^{-1}_f \mathrm{Var}_f(\nabla \ell(\theta_f,X^{(f)}, y^{(f)})))\bigg)\beta\\
     &= \text{tr}(S^{-1}_f \mathrm{Var}_f(\nabla \ell(\theta_f,X^{(f)}, y^{(f)}))) \ \mathcal{C}(\Sigma^{K}).
\end{align*}

Now, let $\Delta^{(k)} = v^{(k)} - v^{(1)}$ and $v^{(k)} = W^{(k)} + \sqrt{c_k} L^{(k)}$, where $L^{(k)}$'s are i.i.d.\ standard Gaussian random variables, independent of $W^{(k)}$. Letting $\beta_1 = 1- \sum_{k>1}\beta_k$, we can therefore write by Lemma \ref{lemma:samp_asymp} and the construction of the random shift model,
$$C(\Sigma^K)\propto \frac{1}{m} \min_{\beta_{2:K}} E[(W^{(1)} - v^{(1)}-\beta^\intercal_{2:K}\Delta^{(2:K)})^2]$$
and 
$$\frac{C(\Sigma^{K+1})}{C(\Sigma^{K})} =\frac{\min_{\beta_{2:K+1}} E[(W^{(1)} - v^{(1)}-\beta^\intercal_{2:K+1}\Delta^{(2:K+1)})^2]}{\min_{\beta_{2:K}} E[(W^{(1)} - v^{(1)}-\beta^\intercal_{2:K}\Delta^{(2:K)})^2]}.$$
Define the residuals of the OLS regression of $W^{(1)} - v^{(1)}$ and $\Delta^{(K+1)}$ on $\Delta^{(2:K)}$ respectively as $r_1$ and $r_{K+1}$.
First, note that we can rewrite the $K+1$ coordinate of the OLS regression of $W^{(1)} - v^{(1)}$ on $\Delta^{(2:K+1)}$  as $\beta^*_{K+1}=\frac{\mathrm{Cov}(r_1,r_{K+1})}{\mathrm{Var}(r_{K+1})}$. Using this, we can write
\begin{align}
     &\min_{\beta_{2:K+1}} E[(W^{(1)} - v^{(1)}-\beta^\intercal_{2:K+1}\Delta^{(2:K+1)})^2]  \nonumber\\
     &= \min_{\beta_{K+1}} \min_{\beta_{2:K}} E[(W^{(1)} - v^{(1)}-\beta^\intercal_{2:K+1}\Delta^{(2:K+1)})^2]   \nonumber\\
     &= \min_{\beta_{K+1}} E\bigg[\bigg( r_1 - \beta_{K+1}r_{K+1}\bigg)^2\bigg]\nonumber\\
    &=\mathrm{Var}(r_1) -2\beta^*_{K+1} E[r_1 r_{K+1}] + {\beta^*}_{K+1}^2\mathrm{Var}(r_{K+1})\nonumber\\
    &=\mathrm{Var}(r_1) -2\frac{\mathrm{Cov}^2(r_1, r_{K+1})}{\mathrm{Var}(r_{K+1})} + \frac{\mathrm{Cov}^2(r_1, r_{K+1})}{\mathrm{Var}^2(r_{K+1})}\mathrm{Var}(r_{K+1})\nonumber\\
     &= \mathrm{Var}(r_1)\bigg(1 -\frac{\mathrm{Cov}^2(r_1, r_{K+1})}{\mathrm{Var}(r_1)\mathrm{Var}(r_{K+1})}\bigg)\nonumber\\
     &=\min_{\beta_{2:K}} E[(W^{(1)} - v^{(1)}-\beta^\intercal_{2:K}\Delta^{(2:K)})^2]\bigg(1 -\frac{\mathrm{Cov}^2(r_1, r_{K+1})}{\mathrm{Var}(r_1)\mathrm{Var}(r_{K+1})}\bigg) \label{pf:pop_pcorr}
\end{align}

Since $L^{(k)}$'s are i.i.d standard Gaussian random variables, independent of $W^{(k)}$, we get that 
\begin{align*}
    \mathrm{Var}(r_1) &= \mathrm{Var}\left(W^{(1)}-v^{(1)}-\sum_{k=2}^{K} \beta_k^* (v^{(k)}-v^{(1)})\right)\\
    &= \mathrm{Var}\left(W^{(1)}-\sum_{k=1}^{K} \beta_k^* W^{(k)}-\sum_{k=1}^{K} \beta_k^*\sqrt{c_k}L^{(k)} \right)\\
    &= \mathrm{Var}\left(W^{(1)}-\sum_{k=1}^{K} \beta_k^* W^{(k)} \right) + \sum_{k=1}^{K} (\beta_k^*)^2 c_k\\
    \mathrm{Var}(r_{K+1})  &=  \mathrm{Var}\left(v^{(K+1)} - v^{(1)}-\sum_{k=1}^{K} \beta_k' (v^{(k)}-v^{(1)})\right)\\
    &=  \mathrm{Var}\left(W^{(K+1)} + \sqrt{c_{K+1}}L^{(K+1)} -\sum_{k=1}^{K} \beta_k' W^{(k)}- \sum_{k=1}^{K} \beta_k'\sqrt{c_k}L^{(k)} \right)\\
    &=  \mathrm{Var}\left(W^{(K+1)} -\sum_{k=1}^{K} \beta_k' W^{(k)} \right) + \sum_{k=1}^{K} (\beta_k')^2 c_k   + c_{K+1}\\
    \mathrm{Cov}(r_1, r_{K+1}) &= \mathrm{Cov}(W^{(1)} -v^{(1)} - \sum_{k=2}^{K} \beta_k^* (v^{(k)}-v^{(1)}), v^{(K+1)} -v^{(1)} - \sum_{k=2}^{K} \beta_k' (v^{(k)}-v^{(1)}))\\
    &=\mathrm{Cov}(W^{(1)} - \sum_{k=1}^{K} \beta_k^* W^{(k)}, W^{(K+1)} -\sum_{k=1}^{K} \beta_k' W^{(k)}) + \sum_{k=1}^K \beta_k^* \beta_k' c_k,
\end{align*}
giving us the result  

\begin{equation*}
 \lim_{G \rightarrow \infty} \lim_{m,n_k \rightarrow \infty}  \frac{\mathcal{E}^{K+1}}{\mathcal{E}^{K}}  = 1 - \frac{(\mathrm{Cov}( W^{(1)} -  W^{\beta^*}, W^{(K+1)} - W^{\beta'}) +  \sum_{k=1}^K \beta_k^*\beta_k' c_k )^2 }{ (\mathrm{Var}(W^{(1)} - W^{\beta^*}) + \sum_{k=1}^K (\beta_k^*)^2 c_k )  ( \mathrm{Var}(W^{(K+1)} - W^{\beta'}) + \sum_{k=1}^K (\beta_k')^2 c_k + c_{K+1} ) }.
\end{equation*}

\end{proof}

\subsection{Proof of Corollary ~\ref{corr: indepW}}\label{app: corr1}

\begin{proof}

In this proof, we use the notation of the proof of Theorem~\ref{thm: pred_rho}.
Let $\sigma^2_k \coloneqq \mathrm{Var}(W^{(k)})+c_k$ and $W^{(1)} \stackrel{\text{a.s.}}{=} 1$. 
We have under independent $W^{(k)}$'s and $E[W^{(k)}] = 1$ that
\begin{equation*}
   \mathrm{Var}(r_{K+1}) =  \min_{{\beta}^\intercal 1 = 1} \mathrm{Var}(W^{(K+1)}) + \sum_{k=1}^K \text{Var}(W^{(k)}) \beta_k^2  + c_k \beta_k^2
\end{equation*}
giving the first order condition under constraint ${\beta}^\intercal 1 = 1$ of 
\begin{align*}
    0&=2\beta'_k \sigma_k^2 - 2  \beta'_j \sigma_j^2,
\end{align*}
resulting in optimal $\beta_k' = \frac{1/\sigma_k^2}{\sum_{j=1}^K 1/\sigma_j^2}$. Therefore, we get
\begin{align}\label{corr1_app: ek+1}
    \mathrm{Var}(r_{K+1}) &= \mathrm{Var}(W^{(K+1)}) + \left(\sum_{k=1}^K 1/\sigma^2_k\right)^{-1}.
\end{align} 
Similarly, we have
\begin{align*}
    \mathrm{Var}(r_1) &= \min_{{\beta}^\intercal 1 = 1} \text{Var}(W^{(1)} - \sum_{k=1}^K \beta_k W^{(k)}) + \sum_{k=1}^K c_k \beta_k^2
\end{align*}
Recall that by assumption $W^{(1)} = 1$. Thus,
\begin{align*}
    \mathrm{Var}(r_1) &= \min_{{\beta}^\intercal 1 = 1}  \sum_{k=1}^K \beta^2_k \sigma_k^2
\end{align*}
Solving this equation gives us the same optimal $\beta^*_k =\frac{1/\sigma_k^2}{\sum_{j=1}^K 1/\sigma_j^2}$.
Evaluating the variance at the optimal $\beta^*$ values and simplifying gives 
\begin{align}\label{corr1_app: e1}
    \mathrm{Var}(r_1) &=  \left(\sum_{k=1}^K 1/\sigma^2_k\right)^{-1}.
\end{align}
Furthermore, we have
\begin{align}\label{corr1_app: num}
  \mathrm{Cov}(r_1, r_{K+1}) = \left(\sum_{k=1}^{K}(\beta_k^*)^2\sigma_k^2\right)^2 = \left(\sum_{k=1}^K 1/\sigma_k^2\right)^{-2}.
\end{align}
Combining Eqns \ref{corr1_app: ek+1}, \ref{corr1_app: e1}, \ref{corr1_app: num} and rearranging gives 
    \begin{align*}
        \rho^2_{K+1|1, \ldots,K} = \frac{\mathrm{Cov}^2(r_1, r_{K+1})}{\mathrm{Var}(r_1)\mathrm{Var}(r_{K+1})}  &=\frac{\left(\sum_{k=1}^K 1/\sigma_k^2\right)^{-2} }{\left(\sum_{k=1}^K 1/\sigma^2_k\right)^{-1}\left(\sigma_{K+1}^2+(\sum_{k=1}^K 1/\sigma^2_k)^{-1}\right)}\\
        &= \frac{1/\sigma^2_{K+1}}{((\sum_{k=1}^{K} \frac{1}{\sigma_k^2}) \sigma^2_{K+1}+1)(1/\sigma^2_{K+1})} \\
        &= \frac{(\mathrm{Var}(W^{(K+1)})+c_{K+1})^{-1}}{\sum_{k=1}^{K+1} (\mathrm{Var}(W^{(k)})+c_{k})^{-1}}.
    \end{align*}
\end{proof}

\subsection{Proof of Proposition~\ref{prop: consistency_cov} }\label{app: consistency_cov} 

\begin{proof}
From Eqn \ref{pf:pop_pcorr}, we have by definition of partial correlation that equivalently, 

\begin{equation*}
 \lim_{G \rightarrow \infty} \lim_{m,n_k \rightarrow \infty}  \frac{\mathcal{E}^{K+1}}{\mathcal{E}^{K}}  = 1 - \text{partial.correlation}^2(W^{(1)}-v^{(1)},\Delta^{(K+1)} | \Delta^{(2:K)} ) .
\end{equation*}
So it remains to show that the empirical partial correlation of $Z^{(1)}_{l=1,\ldots,L}$ and $Z^{(K+1)}_{l=1,\ldots,L}$ given $Z^{(2:K)}_{l=1,\ldots,L}$ converges in distribution to the partial correlation of $W^{(1)}-v^{(1)}$ and $\Delta^{(K+1)} $ given $ \Delta^{(2:K)} )$.

By the distributional CLT (Lemma \ref{lemma:samp_asymp}), for each fixed $L$, and for $m,n_k \rightarrow \infty$,
\begin{equation}
   \sqrt{m} Z^{(1:K+1)} \stackrel{d}{\rightarrow} J,
\end{equation}
where $J$ is a matrix with independent rows, and each row is a draw from a Gaussian distribution with mean zero and covariance matrix $\Sigma^K= A\Sigma^{W}A^\intercal + \mathrm{Diag}(c_1, c_2,\ldots,c_K)$. Since correlation is invariant to scaling, the empirical partial correlation of $Z^{(1)}_{l=1,\ldots,L}$ and $Z^{(K+1)}_{l=1,\ldots,L}$ given $Z^{(2:K)}_{l=1,\ldots,L}$ is equal to the empirical correlation of $\sqrt{m} Z^{(1)}_{l=1,\ldots,L}$ and $\sqrt{m} Z^{(K+1)}_{l=1,\ldots,L}$ given $\sqrt{m} Z^{(2:K)}_{l=1,\ldots,L}$. Furthermore, by continuous mapping, the distributional CLT  implies that the latter empirical partial correlation asymptotically ($m \rightarrow \infty$) has the same distribution as the empirical partial correlation of $J_{1,l}$ and $J_{K+1,l}$ given $J_{2:K,l}$ over $l=1,\ldots,L$. Let $\hat \rho_{\text{Gaussian},L}$ denote the empirical correlation of $J_{1,l}$ and $J_{K+1,l}$ given $J_{2:K,l}$ over $l=1,\ldots,L$. Formally, for fixed $L$ and $m \rightarrow \infty$
\begin{equation*}
    \hat \rho_{K+1|1,2,\ldots,K} \stackrel{d}{\rightarrow} \hat \rho_{\text{Gaussian},L}.
\end{equation*}
Thus,
\begin{equation*}
    \lim_{m \rightarrow \infty} 
 P ( \sqrt{L}\frac{\hat{\rho}_{K+1|1,2, \ldots,K} - \rho_{K+1|1,2, \ldots,K}}{( 1- \rho^2_{K+1|1,2, \ldots,K})} \le x ) = P( \sqrt{L} \frac{ \rho_{\text{Gaussian},L} - \rho_{K+1|1,2,\ldots,K}}{ ( 1- \rho^2_{\text{Gaussian},L})} \le x )
\end{equation*}
A well-known application of the delta method gives us that for the sample correlation of i.i.d.\ Gaussians $(J_l)_{l=1,\ldots,L}$, we have
\begin{equation*}
    \sqrt{L} (\rho_{\text{Gaussian},L} - \rho_{K+1|1,2, \ldots,K} )\rightarrow N(0, (1- \rho^2_{\text{Gaussian},L})^2), 
\end{equation*}
and therefore
\begin{equation*}
    \sqrt{L} (\rho^2_{\text{Gaussian},L} - \rho^2_{K+1|1,2, \ldots,K} )\rightarrow N(0, 4\rho^2_{\text{Gaussian},L}(1- \rho^2_{\text{Gaussian},L})^2). 
\end{equation*} Thus,
\begin{equation*}
  \lim_{L \rightarrow \infty}  \lim_{m \rightarrow \infty} 
 P ( \sqrt{L}\frac{\hat{\rho}^2_{K+1|1,2, \ldots,K} - \rho^2_{K+1|1,2, \ldots,K}}{\hat {\sigma}^2} \le x ) = \lim_{L \rightarrow \infty} P( \sqrt{L} \frac{ \rho^2_{\text{Gaussian},L} - \rho^2_{K+1|1,2, \ldots,K}}{ 4\rho_{\text{Gaussian},L}( 1- \rho^2_{\text{Gaussian},L})^2} \le x ) = \Phi(x).
\end{equation*}
This completes the proof of Proposition \ref{prop: consistency_cov}. Note that this also implies we can consistently estimate $\Sigma^K$ using the empirical covariance matrix of $Z^{1:K}$. 

\end{proof}

\section{Proofs of supporting results}

\subsection{Lemma \ref{lemma:samp_asymp}}\label{app: samp_asymp}

\begin{lemma}[Distributional CLT]\label{lemma:samp_asymp} 
Let $\phi_k: \mathcal{D} \rightarrow \mathbbm{R}$ be any Borel measurable square-integrable function under $P_{f}$ and consider the setting where distributional and sampling uncertainty are of the same order, i.e., $m/n_k \rightarrow c_k > 0$ for each $k=1, \ldots, K$. Then 
$$
\sqrt{m}
\begin{pmatrix}
\begin{pmatrix}
\frac{1}{n_1} \sum_{i=1}^{n_1} \phi_1(D^{(1)}_i) \\
\vdots \\
\frac{1}{n_K} \sum_{i=1}^{n_K} \phi_K(D_i^{(K)})
\end{pmatrix}
-
\begin{pmatrix}
E_1[\phi_1(D^{(1)})] \\
\vdots \\
E_1[\phi_K(D^{(K)})]
\end{pmatrix}
\end{pmatrix}
\overset{d}{\rightarrow} N\big(0, \Sigma^K \odot \mathrm{Var}_f(\phi_1(D^{(f)}), \ldots, \phi_K(D^{(f)}))\big),
$$
where $\Sigma^K$ is the asymptotic covariance matrix defined as
\begin{equation*}
    \Sigma^K =   A\Sigma^{W} A^\intercal+ \mathrm{Diag}(c_1, c_2,\ldots,c_K),
\end{equation*}
for 
\begin{equation*}
    A = \begin{pmatrix}
     0_{1 \times 1} & 0_{1 \times (K-1)} \\
-1_{(K-1) \times 1} & I_{(K-1) \times (K-1)}
    \end{pmatrix}.
\end{equation*}
Here, $\Sigma^{W} \in \mathbb{R}^{K \times K}$ has entries $\operatorname{Cov}(W^{(k)}, W^{(k')})$  and $\odot$ denotes the Hadamard product, i.e., under the above assumptions,

\begin{equation*}
    \Sigma^K =    \begin{pmatrix}
      c_1 & 0 & \ldots & 0\\
      0 & \Sigma^W_{1,1} + \Sigma^W_{2,2}-2\Sigma^W_{1,2}+c_2 & \ldots & \Sigma^W_{1,1} + \Sigma^W_{2,K} - \Sigma^W_{1,2} - \Sigma^W_{1,K}\\
      \vdots & \vdots & & \vdots\\
      0 & \Sigma^W_{1,1} + \Sigma^W_{K,2} - \Sigma^W_{2,1} - \Sigma^W_{K,1} & \ldots & \Sigma^W_{K,K} + \Sigma^W_{1,1}-2\Sigma^W_{1,K} +c_K
    \end{pmatrix}
\end{equation*}
\end{lemma}

\begin{proof}

In general, we follow the proof of Theorem 1 of \cite{random_shift2}, but make adjustments due to the different asymptotic regimes. In our setting, sampling and distributional uncertainty are of the same order, i.e. $n_{k,m}$ are sequences of of natural numbers such that $m/n_k$ converge to positive real numbers $c_k$, and the target distribution is random. 

It follows from auxiliary Lemma 2 of \cite{random_shift2} that the asymptotic distribution of $E_1[\phi_1(D^{(1)})], \ldots, E_K[\phi_K(D^{(K)})]$ centered around the fixed means, $E_f[\phi_1(D^{(f)})], \ldots, E_f[\phi_K(D^{(f)})]$ is 
$$
\sqrt{m}
\begin{pmatrix}
\begin{pmatrix}
E_1[\phi_1(D^{(1)})] \\
\vdots \\
E_K[\phi_K(D^{(K)})]
\end{pmatrix}
-
\begin{pmatrix}
E_f[\phi_1(D^{(f)})] \\
\vdots \\
E_f[\phi_K(D^{(f)})]
\end{pmatrix}
\end{pmatrix}
\overset{d}{\rightarrow} N\big(0, \Sigma^W \odot V),
$$
where $V\coloneqq \mathrm{Var}_f(\phi_1(D^{(f)}), \ldots, \phi_K(D^{(f)}))$ is the covariance matrix of $(\phi_1(D^{(f)}), \ldots, \phi_K(D^{(f)}))$.
Then letting
\begin{equation*}
    A = \begin{pmatrix}
     0_{1 \times 1} & 0_{1 \times (K-1)} \\
-1_{(K-1) \times 1} & I_{(K-1) \times (K-1)}
    \end{pmatrix},
\end{equation*}
we get by applying the mixed-product property of the Kronecker product and continuous mapping, the asymptotic distribution 
\begin{align}\label{eqn: extended_prop1}
&\sqrt{m}
\begin{pmatrix}
\begin{pmatrix}
E_1[\phi_1(D^{(1)})] \\
\vdots \\
E_K[\phi_K(D^{(K)})]
\end{pmatrix}
-
\begin{pmatrix}
E_1[\phi_1(D^{(1)})] \\
\vdots \\
E_1[\phi_K(D^{(1)})]
\end{pmatrix}
\end{pmatrix}
\overset{d}{\rightarrow} N\big(0,  A \Sigma^W A^\intercal  \odot V).
\end{align}
Note that the top row of the asymptotic covariance matrix of Eqn \ref{eqn: extended_prop1} is $0$.

We first prove the result for bounded functions $\phi^B$. Define $V^B\coloneqq \mathrm{Var}_f(\phi^B_1(D^{(f)}), \ldots, \phi^B_K(D^{(f)}))$ as the covariance matrix of the bounded functions $(\phi^B_1(D^{(f)}), \ldots, \phi^V_K(D^{(f)}))$. Let $\alpha_k$, $k=1,\ldots,K$, be arbitrary real numbers and recall that $\hat{E}_k[\phi_k(D^{(k)})] \coloneqq \frac{1}{n_k}\sum_{i=1}^{n_k} \phi_k(D^{(k)}_i)$. Without loss of generality, assume $E_f[\phi_k(D^{(f)})]=0$ for all $k$. We first analyze the behavior of the following linear combination of sample mean terms 
$$\sqrt{m}\left(\sum_{k=1}^K \alpha_k(\hat{E}_k[\phi_k^B(D^{(k)})]-(E_1[\phi_1^B(D^{(1)})]-E_f[\phi^B_k(D^{(f)})]))\right)= \sqrt{m}\left(\sum_{k=1}^K \alpha_k (\hat{E}_k[\phi_k^B(D^{(k)})]-E_1[\phi_1^B(D^{(1)})])\right).$$
Observe that
\begin{align*}
    &P\bigg(\sqrt{m}\bigg(\sum_{k=1}^K \alpha_k (\hat{E}_k[\phi_k^B(D^{(k)})]-E_1[\phi_k^B(D^{(1)})])\bigg) \leq x \bigg) \\
    &= P\bigg(\sqrt{m}\bigg(\sum_{k=1}^K \alpha_k(\hat{E}_k[\phi_k^B(D^{(k)})] - E_k[\phi_k^B(D^{(k)})]) + \sum_{k=1}^K \alpha_kE_k[\phi_k^B(D^{(k)})]\bigg) \leq x+\sqrt{m}\bigg(\sum_{k=1}^K \alpha_k E_1[\phi_1^B(D^{(1)})]\bigg) \bigg)\\
    &=E\left[P\bigg(\sum_{k=1}^K \alpha_k(\hat{E}_k[\phi_k^B(D^{(k)})] - E_k[\phi_k^B(D^{(k)})]) \leq  m^{-1/2} x- \sum_{k=1}^K \alpha_k(E_k[\phi_k^B(D^{(k)})]-E_1[\phi_k^B(D^{(1)})])\bigg| W\bigg)\right],
\end{align*}
where $W$ are the random weights defined in Section \ref{sec: setting}. Now, to evaluate the probability term, we note that conditional on $W$, $\{\{\phi_k^B(D^{(k)}_i)\}_{i=1}^{n_k}\}_{k=1}^K$ are independent so we can apply the Berry-Esseen bound. 

For clarity, we introduce new notation before applying Berry-Esseen. In the following, $\gamma(\alpha_{1:K}, B, K)>1$ denotes a finite constant depending on $\alpha_1, \ldots, \alpha_K, B, K$ and $\Phi(\mathbf{x})$ denotes the standard Gaussian CDF. Denote 
\begin{align*}
    T_i^{(k)} \coloneqq  \alpha_k\bigg(\frac{1}{n_k}(\phi_k^B(D^{(k)}_i) - E_k[\phi_k^B(D^{(k)}_i)] \bigg) \quad \text{and} \quad S_{n_k,K} \coloneqq \sum_{k=1}^K  \sum_{i=1}^{n_k} T_i^{(k)},
\end{align*} 
where $E[T_i^{(k)}|W]=0$ since $E[\phi_k^B(D^{(k)}_i)|W] = E_k[\phi_k^B(D^{(k)}_i)]$ by construction of the random shift model and $$\sigma_{i,k}^2 \coloneqq \mathrm{Var}(T_i^{(k)}|W)=\alpha_k^2/n_k^2\mathrm{Var}(\phi_k^B(D_i^{(k)})|W)=\alpha_k^2/n_k^2\mathrm{Var}_k(\phi_k^B(D_i^{(k)})).$$ We also have that the third moment of $T_i^{(k)}$ conditional on $W$ is
 \begin{align*}
 E[|T_i^{(k)}|^3|W] &=\alpha_k^3/n_k^3 E\left[\left|\phi_k^B(D^{(k)}_i) - E_k[\phi_k^B(D^{(k)}_i)]\right|^3 \bigg| W\right].
 \end{align*}
Therefore, applying Berry-Esseen gives
\begin{align*}
    \sup_x \bigg|P\left(\frac{S_{n_k,K}}{\sqrt{\sum_{k=1}^K  \sum_{i=1}^{n_k}  \sigma_{i,k}^2}} \leq x \bigg| W\right) -\Phi(x)\bigg| &\leq \frac{\gamma(\alpha_{1:K}, B, K)\sum_{k=1}^K\sum_{i=1}^{n_k}E[|T_i^{(k)}|^3|W]}{(\sum_{k=1}^K\sum_{i=1}^{n_k}  \sigma_{i,k}^2)^{3/2}}\\
    &= \frac{\gamma(\alpha_{1:K}, B, K) \sum_{k=1}^K\alpha_k^3/n_k^2 E\left[\left|\phi_k^B(D^{(k)}) - E_k[\phi_k^B(D^{(k)})]\right|^3 \bigg| W\right]}{(\sum_{k=1}^K  \alpha_k^2/n_k\mathrm{Var}_k(\phi_k^B(D^{(k)})))^{3/2}}\\
    &= O(1/\sqrt{m}),
\end{align*} 
since $K$ is a finite constant and  $m=O(n_k)$ for any $k$ to obtain $O(1/\sqrt{m})$ on the upper bound. Therefore,

\begin{align*}
 &\limsup_{m,n_k \rightarrow \infty} P\bigg(\sqrt{m}\bigg(\sum_{k=1}^K \alpha_k (\hat{E}_k[\phi_k^B(D^{(k)})]-E_1[\phi_k^B(D^{(1)})])\bigg) \leq x \bigg)\\
    &= \limsup_{m,n_k \rightarrow \infty} E\left[\Phi\left(\left(m\sum_{k=1}^K  \alpha_k^2/n_k\mathrm{Var}_k(\phi_k^B(D^{(k)}))\right)^{-1/2}\left(x- \sqrt{m}\sum_{k=1}^K \alpha_k(E_k[\phi_k^B(D^{(k)})]-E_1[\phi_k^B(D^{(1)})] \right)\right) \right] + O(1/ \sqrt{m})
\end{align*}
Let $G, G'$ be standard Gaussian random variables with $G$ independent of $G'$. We have by Eqn~\eqref{eqn: extended_prop1} that $$\sqrt{m}\bigg(\sum_{k=1}^K \alpha_k (E_k[\phi^B(D^{(k)})] - E_1[\phi^B(D^{(1)})])\bigg) \overset{d}{\rightarrow} (\alpha^\intercal (A \Sigma^W A^\intercal \odot V^B)\alpha)^{1/2}G.$$ Since $\Phi$ is a continuous, bounded function, we interchange the limits and expectation to get as $m,n_k \rightarrow \infty$, 
\begin{align*}
&E\left[\Phi\left(\left(m\sum_{k=1}^K  \alpha_k^2/n_k\mathrm{Var}_k(\phi_k^B(D^{(k)}))\right)^{-1/2}\left(x- \sqrt{m}\sum_{k=1}^K \alpha_k(E_k[\phi_k^B(D^{(k)})]-E_1[\phi_k^B(D^{(1)})] \right)\right) \right] \\
  &\rightarrow E\bigg[P\bigg(G' \leq  \left(\sum_{k=1}^K  \alpha_k^2c_k\mathrm{Var}_f(\phi_k^B(D^{(f)}))\right)^{-1/2}\left(x  -(\alpha^\intercal (A \Sigma^W A^\intercal\odot V^B) \alpha)^{1/2}G \right) \bigg)\bigg] \\
  &=P\bigg(\left(\sum_{k=1}^K  \alpha_k^2c_k\mathrm{Var}_f(\phi_k^B(D^{(f)}))\right)^{1/2}G'+ (\alpha^\intercal (A \Sigma^W A^\intercal\odot V^B) \alpha)^{1/2}G \leq x \bigg).
\end{align*}

Therefore, we have in summary that 
\begin{equation}\label{eqn: bdd_distr}
     \sqrt{m}\bigg(\sum_{k=0}^K \alpha_k (\hat{E}_k[\phi_k^B(D^{(k)})]-E_1[\phi_k^B(D^{(1)})]) \bigg)\overset{d}{\rightarrow} N(0,\alpha^\intercal(\Sigma^K\odot V^B)\alpha), 
\end{equation}
where $\Sigma^K$ is the asymptotic covariance matrix defined as
\begin{equation*}
    \Sigma^K =   A\Sigma^{W} A^\intercal+ \mathrm{Diag}(c_1, c_2,\ldots,c_K).
\end{equation*}
By the Cramér-Wold device, 
$$
\sqrt{m}
\begin{pmatrix}
\begin{pmatrix}
\frac{1}{n_1} \sum_{i=1}^{n_1} \phi^B_1(D^{(1)}_i) \\
\vdots \\
\frac{1}{n_K} \sum_{i=1}^{n_K} \phi^B_K(D^{(K)}_i)
\end{pmatrix}
-
\begin{pmatrix}
E_1[\phi^B_1(D^{(1)})] \\
\vdots \\
E_1[\phi^B_K(D^{(1)})]
\end{pmatrix}
\end{pmatrix}
\overset{d}{\rightarrow} N(0, \Sigma^K \odot V^B).
$$

Now we extend to unbounded functions $\phi_k$. For any $\epsilon > 0$ we can choose $0 < B < \infty $ sufficiently large such that the bounded function
$$\phi_k^B=\phi_k \mathbbm{1}_{ | \phi_k | \le B} + E_{f}[\phi_k | | \phi_k | > B] \mathbbm{1}_{ | \phi | > B},$$ suffices $E_f[|\phi_k(D^{(f)})-\phi_k^B(D^{(f)})|^2] \leq \epsilon$. Note that $E_f[\phi_k(D^{(f)})] = E_f[\phi_k^B(D^{(f)})]$ In the following, we will analyze the behavior of $\hat{E}_k[\phi_k^B(D^{(k)}) -\phi_k(D^{(k)})]$ and $E_1[\phi_k^B(D^{(1)}) -\phi_k(D^{(1)})]$. 

The following follows the proof of Theorem 3.3 of \cite{random_shift3}, with adjustments made for a random target and multiple sources. First, by Chebyshev, we have that conditional on the random shift for each $k$ for any $\epsilon >0$,
\begin{align*}
	P(|\hat{E}_k[\phi(D^{(k)})-\phi^B_k(D^{(k)})] - E_k[\phi(D^{(k)})-\phi^B_k(D^{(k)})]| \geq \epsilon | W^{(k)}) &\leq \frac{\frac{1}{n_k} \mathrm{Var}_k[\phi_k(D^{(k)})-\phi^B_k(D^{(k)})]}{\epsilon^2}\\
    &\leq \frac{\frac{1}{n_k} E_k[(\phi_k(D^{(k)})-\phi^B_k(D^{(k)}))^2]}{\epsilon^2}
\end{align*} 
Taking the expectation over the random shift on both sides, we get
\begin{align}\label{eqn: chebyshev1}
\begin{split}
	P(|\hat{E}_k[\phi(D^{(k)})-\phi^B_k(D^{(k)})] - E_k[\phi(D^{(k)})-\phi^B_k(D^{(k)})] | \geq \epsilon ) &\leq \frac{\frac{1}{n_k} E_f[(\phi_k(D^{(f)})-\phi^B_k(D^{(f)}))^2]}{\epsilon^2}\\
    &=\frac{\frac{1}{n_k} \mathrm{Var}_f[\phi_k(D^{(f)})-\phi^B_k(D^{(f)})]}{\epsilon^2}    
\end{split}
\end{align} 

Recall that $P_f(U\in I_j)=1/m$ and $D=h(U)$. By assumption, there exists a $w_f>0$ such that $W^{(k)}\geq w_f$ for any $k$, so 
\begin{align*}
&\bigg|E_k[\phi_k(D^{(k)})-\phi^B_k(D^{(k)})]-E_1[\phi_k(D^{(1)})-\phi^B_k(D^{(1)})]\bigg|\\
&= \bigg|\sum_{j=1}^m \frac{W^{(k)}_j}{ \sum_{j'} W_{j'}^{(k)}} E_f[\phi_k(h(U))-\phi^B_k(h(U))) | \{U\in I_j\}] -\sum_{j=1}^m \frac{W^{(1)}_j}{ \sum_{j'} W_{j'}^{(1)}} E_f[\phi_k(h(U))-\phi^B_k(h(U))) | \{U\in I_j\}] \bigg|\\ 
&\leq \frac{1}{mw_f}\bigg|\sum_{j=1}^m (W^{(k)}_j-W^{(1)}_j) E_f[\phi_k(h(U))-\phi^B_k(h(U))) | \{U\in I_j\}]\bigg|. 
\end{align*}
Since $W_j^{(k)}$'s are i.i.d.\ across $j$ for each $k$, Chebyshev combined with Jensen's gives
\begin{align} \label{eqn: chebyshev2}
\begin{split}
    &P(|E_k[\phi_k(D^{(k)})-\phi^B_k(D^{(k)})]-E_1[\phi_k(D^{(1)})-\phi^B_k(D^{(1)})]| \geq \epsilon)\\
    &\leq \frac{ \sum_{j=1}^m \mathrm{Var}[W^{(k)}_j-W^{(1)}_j] E_f[(\phi_k(h(U))-\phi^B_k(h(U)) | \{U\in I_j\})^2]}{w_f^2\epsilon^2m^2}\\
    &\leq \frac{m \mathrm{Var}[W^{(k)}_j-W^{(1)}_j] \mathrm{Var}_f[\phi_k(D^{(f)})-\phi^B_k(D^{(f)}))]}{w_f^2\epsilon^2m^2},
    \end{split}
\end{align}
since $E_f[\phi_k(D^{(f)})-\phi^B_k(D^{(f)})]=0$. Combining Eqns~\ref{eqn: chebyshev1} and~\ref{eqn: chebyshev2},
\begin{align*}
    &P(|\hat{E}_k[\phi(D^{(k)})-\phi^B_k(D^{(k)})] - E_1[\phi_k(D^{(1)})-\phi^B_k(D^{(1)})]| \geq \epsilon ) \\
    &\leq P(|\hat{E}_k[\phi(D^{(k)})-\phi^B_k(D^{(k)})] - E_k[\phi_k(D^{(k)})-\phi^B_k(D^{(k)})] | \geq \epsilon/2 ) \\
    &+ P(|E_k[\phi(D^{(k)})-\phi^B_k(D^{(k)})]- E_1[\phi_k(D^{(1)})-\phi^B_k(D^{(1)})]| \geq \epsilon/2)\\
    &\leq \frac{4}{\epsilon^2}\mathrm{Var}_f[(\phi_k(D^{(f)})-\phi^B_k(D^{(f)}))^2]\bigg(\frac{1}{n_k} + \frac{\mathrm{Var}[W^{(k)}_j-W^{(1)}_j] }{w_f^2m} \bigg)
\end{align*}

For any $\epsilon > 0$, $\epsilon' > 0$, there exists $B > 0$ large enough so that $ P( \sqrt{m} |\sum_{k=1}^K \alpha_k(\hat{E}_k[\phi_k(D^{(k)}) - \phi^B_k(D^{(k)})]- E_1[\phi_k(D^{(1)})-\phi^B_k(D^{(1)})])| \geq \epsilon ) \le \epsilon'$. Therefore, for any $\epsilon > 0$ and any $\epsilon'>0$, we can choose $B$ large enough so that for any $m$,
\begin{equation}\label{eqn: abs_bdd}
    P\left( (\alpha^\intercal (\Sigma^K \odot V^B)\alpha)^{-1/2} \sqrt{m} \left|\sum_{k=1}^K \alpha_k(\hat{E}_k[\phi_k(D^{(k)}) - \phi^B_k(D^{(k)})]- E_1[\phi_k(D^{(1)})-\phi^B_k(D^{(1)})]) \right| \geq \epsilon \right) \leq \epsilon'.
\end{equation}
Since $\phi^B(D^{(k)}) \rightarrow \phi(D^{(k)})$ as $B\rightarrow \infty$, $E_f[\phi^B(D^{(f)})^2] \rightarrow E_f[\phi(D^{(f)})^2]$ by dominated convergence. By assumption, $E_f[\phi^B(D^{(f)})]=E_f[\phi(D^{(f)})]=0$ so $\mathrm{Var}_f[\phi^B(D^{(f)})] \rightarrow \mathrm{Var}_f[\phi(D^{(f)})]$. This implies that for any $\omega' >0$  we have $|\sqrt{\mathrm{Var}_f[\phi(D^{(f)})]/\mathrm{Var}_f[\phi^B(D^{(f)})]}-1| < \omega'$ for $B$ large enough. Equivalently,   
$\Sigma^K_m\odot V\leq \Sigma^K_m\odot V^B (1+\omega)^2$  for any $\omega>0$ and $B$ large enough, where $\Sigma^K_m \coloneqq m(A\Sigma^W A^\intercal+ \mathrm{Diag}(1/n_1, \ldots, 1/n_K))$. 
Putting it all together, we have 
\begin{align*}
    &\limsup_{m,n_k \rightarrow \infty} P\bigg((\alpha^\intercal (\Sigma^K_m\odot V) \alpha)^{-1/2} \sqrt{m}\bigg(\sum_{k=1}^K \alpha_k(\hat{E}_k[\phi_k(D^{(k)})]-E_1[\phi_1(D^{(1)})])\bigg) \leq x \bigg)\\
    &\leq \limsup_{m,n_k \rightarrow \infty} P\bigg((\alpha^\intercal (\Sigma^K_m\odot V^B) \alpha)^{-1/2} \sqrt{m}\bigg(\sum_{k=1}^K \alpha_k(\hat{E}_k[\phi_k(D^{(k)})]-E_1[\phi_1(D^{(1)})])\bigg) \leq \max(x, (1+\omega)x) \bigg)\\
    &\leq \limsup_{m,n_k \rightarrow \infty} P\bigg((\alpha^\intercal (\Sigma^K_m\odot V^B) \alpha)^{-1/2} \sqrt{m}\bigg(\sum_{k=1}^K \alpha_k(\hat{E}_k[\phi^B_k(D^{(k)})]-E_1[\phi^B_1(D^{(1)})])\bigg) \leq \max(x, (1+\omega)x) + \epsilon \bigg)\\
    &+  \limsup_{m,n_k \rightarrow \infty} P\left((\alpha^\intercal (\Sigma^K_m\odot V^B) \alpha)^{-1/2} \sqrt{m} \left|\sum_{k=1}^K \alpha_k(\hat{E}_k[\phi_k(D^{(k)}) - \phi^B_k(D^{(k)})]- E_1[\phi_k(D^{(1)})-\phi^B_k(D^{(1)})]) \right| \geq \epsilon \right)\\
    & \leq \Phi(\max(x, (1+\omega)x)+\epsilon)+\epsilon',
\end{align*}
where the last inequality uses Eqns \ref{eqn: bdd_distr} and \ref{eqn: abs_bdd}. 
Since $\epsilon > 0$, $\epsilon'>0$ can be chosen arbitrary small for $B \rightarrow \infty$,
\begin{equation*}
    \limsup_{m,n_k \rightarrow \infty} P\bigg((\alpha^\intercal (\Sigma^K_m\odot V) \alpha)^{-1/2} \sqrt{m}\bigg(\sum_{k=1}^K \alpha_k(\hat{E}_k[\phi_k(D^{(k)})]-E_1[\phi_1(D^{(1)})])\bigg) \leq x \bigg) \leq \Phi(x).
\end{equation*}
By the same argument as above,
\begin{equation}
    \liminf_{m,n_k \rightarrow \infty} P\bigg((\alpha^\intercal (\Sigma^K_m\odot V) \alpha)^{-1/2} \sqrt{m}\bigg(\sum_{k=1}^K \alpha_k (\hat{E}_k[\phi_k(D^{(k)})]-E_1[\phi_1(D^{(1)})])\bigg) \leq x \bigg) \geq\Phi(x).
\end{equation}
Therefore, by the Cramér-Wold device, 
$$
\sqrt{m}
\begin{pmatrix}
\begin{pmatrix}
\frac{1}{n_1} \sum_{i=1}^{n_1} \phi_1(D^{(1)}_i) \\
\vdots \\
\frac{1}{n_K} \sum_{i=1}^{n_K} \phi_K(D^{(K)}_i)
\end{pmatrix}
-
\begin{pmatrix}
E_1[\phi_1(D^{(1)})] \\
\vdots \\
E_1[\phi_K(D^{(1)})]
\end{pmatrix}
\end{pmatrix}
\overset{d}{\rightarrow} N(0, \Sigma^K \odot V).
$$

\end{proof}

\subsection{Consistency of \texorpdfstring{$\hat{\theta}$}{theta}}

\begin{lemma}[Consistency]\label{lemma: mestimator_consistency}

Let $\Theta$ be a compact set and $\ell(\cdot)$ be bounded and continuous in $\theta$. Define for any set of non-negative weights $\beta_k \ge 0$ with $\sum_k \beta_k =1$ the estimator
\begin{equation*}
        \hat{\theta} = \argmin_{\theta \in \Theta} \ \sum_{k=1}^K \beta_k \frac{1}{n_k} \sum_{i=1}^{n_k} \ell(\theta,X_i^{(k)}, y_i^{(k)}),
\end{equation*}
and the unique minimizer of the population risk under $P_f$,
\begin{equation*}
    \theta_f = \argmin_{\theta \in \Theta} \ E_f[\ell(\theta,X^{(f)}, y^{(f)})].
\end{equation*}
Then as $m, n_k \rightarrow \infty$, where $m/n_k \rightarrow c_k > 0$, $\hat{\theta}\overset{p}{\rightarrow} \theta_f$. 

Furthermore, if $\ell(\theta,X,y)$ is twice-differentiable, $E_f[\nabla^2\ell(\theta,X^{(f)}, y^{(f)})]$ is L-Lipschitz in an $\epsilon$-ball around $\theta_f$, and $E_f[\nabla^2\ell(\theta,X^{(f)}, y^{(f)})]^{-1}$ is integrable with the third order partial derivatives of $\ell(\theta, X,y)$ dominated by a fixed function $h(\cdot)$ in an $\epsilon'$-ball around $\theta_f$, we have $\hat{\theta} = \theta_f + O_p(1/\sqrt{n_k})$. 
\end{lemma}
\begin{proof}
 First, we note that we have pointwise convergence for any fixed $\theta$, 
\begin{equation}\label{not_lln}
    \frac{1}{n_k} \sum_{i=1}^{n_k} \ell(\theta,X_i^{(k)}, y_i^{(k)}) -E_f[\ell(\theta,X^{(f)}, y^{(f)})] = O_p(1/\sqrt{n_k}),
\end{equation}
by Lemma~\ref{lemma:samp_asymp}, rather than by Law of Large Numbers since $P_k$ is not fixed. By assumption of $\Theta$ compact and $\ell(\cdot)$ continuous in $\theta$ with $\frac{1}{n_k} \sum_{i=1}^{n_k} \ell(\theta,X_i^{(k)}, y_i^{(k)}) \overset{p}{\rightarrow} E_f[\ell(\theta,X^{(f)}, y^{(f)})]$, $\ell(\cdot)$ belongs to the Glivenko-Cantelli class by Example 19.8 of \cite{vdv}. Therefore, we have that 
\begin{equation*}
\max_{k=1,\ldots,K}    \sup_{\theta \in \Theta} \left| \frac{1}{n_k} \sum_{i=1}^{n_k} \ell(\theta,X_i^{(k)}, y_i^{(k)}) -E_f[\ell(\theta,X^{(f)}, y^{(f)})] \right| \overset{p}{\rightarrow} 0.
\end{equation*}
Furthermore, by assumption of $\theta_f$ being a unique minimizer of the population risk, we have the well-separation condition, $\inf_{\theta: d(\theta, \theta_f) \geq \epsilon} E_f[\ell(\theta,X^{(f)}, y^{(f)})] > E_f[\ell(\theta_f,X^{(f)}, y^{(f)})]$. The remainder of the proof follows by the proof of consistency of M-estimators as in Theorem 5.7 of \cite{vdv}. 

By uniform convergence, we have that $ \sum_k \beta_k \frac{1}{n_k} \sum_{i=1}^{n_k} \ell(\theta_f,X_i^{(k)}, y_i^{(k)}) = E_f[\ell(\theta_f, X^{(f)}, y^{(f)})] +o_p(1)$ and since $\theta_f$ is the minimizer of the population risk, $ \sum_k \beta_k \frac{1}{n_k} \sum_{i=1}^{n_k} \ell(\hat{\theta},X_i^{(k)}, y_i^{(k)}) \leq \sum_k \beta_k \frac{1}{n_k} \sum_{i=1}^{n_k} \ell(\theta_f,X_i^{(k)}, y_i^{(k)}) + o_p(1)$. Therefore, we have $\sum_k \beta_k \frac{1}{n_k} \sum_{i=1}^{n_k} \ell(\hat{\theta},X_i^{(k)}, y_i^{(k)}) \leq  E_f[\ell(\theta_f,X^{(f)}, y^{(f)})] + o_p(1)$ and 
\begin{align*}
    E_f[\ell(\hat{\theta},X^{(f)}, y^{(f)})] - E_f[\ell(\theta_f,X^{(f)}, y^{(f)})] &\leq E_f[\ell(\hat{\theta},X^{(f)}, y^{(f)})]- \sum_k \beta_k \frac{1}{n_k} \sum_{i=1}^{n_k} \ell(\hat{\theta},X_i^{(k)}, y_i^{(k)}) +o_p(1)\\
    &\leq \max_k \sup_\theta \left|\frac{1}{n_k} \sum_{i=1}^{n_k} \ell(\theta,X_i^{(k)}, y_i^{(k)})-E_f[\ell(\theta,X^{(f)}, y^{(f)})]\right| + o_p(1)\\
    &=  o_p(1),
\end{align*}
by uniform convergence. By the well-separability condition, for every $\epsilon >0$ there exists a number $\eta>0$ such that $E_f[\ell(\theta, X^{(f)}, y^{(f)})] > E_f[\ell(\theta_f, X^{(f)}, y^{(f)})] + \eta$  for every $\theta$ with $d(\theta, \theta_f) \geq \epsilon$. This implies that the event $\{d(\hat{\theta}, \theta_f) \geq \epsilon\} \subseteq \{E_f[\ell(\hat{\theta}, X^{(f)}, y^{(f)})] > E_f[\ell(\theta_f, X^{(f)}, y^{(f)})]+\eta\}$, so 
\begin{equation*}
    P(d(\hat{\theta}, \theta_f) \geq \epsilon) \leq P(E_f[\ell(\hat{\theta}, X^{(f)}, y^{(f)})] > E_f[\ell(\theta_f, X^{(f)}, y^{(f)})]+\eta) \overset{p}{\rightarrow} 0.
\end{equation*}
Furthermore, by assumption,
$\|\hat{E}_{k}[\nabla^3 \ell(\theta_f,X^{(k)}, y^{(k)})]\| \leq \frac{1}{n_k}\sum_{i=1}^{n_k}h(D_{i}^{(k)}) = O_p(1)$
so by Section 5.3 of \cite{vdv}, a Taylor expansion of $\frac{1}{n_k} \sum_{i=1}^{n_k} \ell(\theta,X_i^{(k)}, y_i^{(k)})$ around $\theta_f$ combined with Lemma~\ref{lemma:samp_asymp} gives us that $\hat{\theta}=\theta_f+O_p(1\sqrt{n_k})$.
\end{proof}

\begin{corollary}\label{cor: mestimator_consistency_theta0}

Define $\hat{\theta}$ as in Lemma~\ref{lemma: mestimator_consistency}
and the unique minimizer of the target population risk under $P_k$
\begin{equation*}
    \theta_k = \argmin_{\theta \in \Theta} \ E_k[\ell(\theta,X^{(k)}, y^{(k)})].
\end{equation*}
Under the setting of Theorem~\ref{thm: pred_rho}, $\hat{\theta} = \theta_k +O_p(1/\sqrt{m})$. 
\end{corollary}

\begin{proof}
We first note that by the same argument as Lemma~\ref{lemma: mestimator_consistency}, we have that $\theta_k = \theta_f + O_p(1/\sqrt{m})$ since $E_k[\ell(\theta, X^{(k)}, y^{(k)})] -E_f[\ell(\theta,X^{(f)}, y^{(f)})] = O_p(1/\sqrt{m})$ by auxiliary Lemma 2 of \cite{random_shift2}.  Then we have that $$\hat{\theta} - \theta_k = \theta_f+O_p(1/\sqrt{n_k})-\theta_f-O_p(1/\sqrt{m}) = O_p(1/\sqrt{m})$$ as $m,n \rightarrow \infty$ since $m=O(n_k)$.
\end{proof}

\subsection{Proof of Remark \ref{confint}}\label{app: ci_transf}
\begin{proof}
The function $g(\rho^2)$ is the variance stabilizing transformation given by
\begin{align*}
    g(\rho^2) = \int_{0}^{\rho^2} \frac{1}{2\sqrt{x}(1-x)}dx = \frac{1}{2}\log \left(\frac{1+\rho}{1-\rho}\right)
\end{align*}
Therefore, by the delta method, $\sqrt{L} (g^{-1}(\rho^2_{\text{Gaussian},L}) - g^{-1}(\rho^2_{K+1|1,2, \ldots,K}) )\rightarrow N(0, 1)$,
and asymptotic validity follows for $$CI \coloneqq [g^{-1}(g(\hat{\rho}^2_{K+1|1, \ldots,K})-L^{-1/2}z_{1-\frac{\alpha}{2}}), g^{-1}(g(\hat{\rho}^2_{K+1|1, \ldots,K})+L^{-1/2}z_{1-\frac{\alpha}{2}})]$$ since $\rho \mapsto g(\rho)$ is monotone:
\begin{align*}
    P(\rho^2_{K+1|1,2, \ldots,K}) \in CI) &\geq P(g^{-1}(\rho^2_{K+1|1,2, \ldots,K}) \in [g^{-1}(\rho^2_{\text{Gaussian},L}) \pm L^{-1/2}z_{1-\frac{\alpha}{2}}])\\
    &=P(|L^{-1/2}(g^{-1}(\rho^2_{\text{Gaussian},L}) - g^{-1}(\rho^2_{K+1|1,2, \ldots,K}))|\leq z_{1-\frac{\alpha}{2}}) \\
    &\rightarrow 1-\alpha.
\end{align*}
This completes the proof of Remark \ref{confint}.

\end{proof}

\section{Optimal sampling under constraints}\label{sec: optimal_samp}

In Section \ref{sec: partial_cor}, we discussed how much an additional data set would reduce excess risk on a target distribution. Here, we derive optimal sampling strategies under various constraints following a similar framework.

\subsection{Optimal sampling under size constraint}\label{sec: size_constraint}

We layer a total sample size constraint to our problem of data collection in the following illustration.  Suppose we can only recruit a total of $N$ 
patients for a clinical trial from $K$ different populations. How can we optimally sample from  $K$ populations? Formally, we want to minimize the following excess risk,
\begin{align}\label{eqn: sample_constraint_obj}
\begin{split}
     n_1^*,\ldots,n_K^*, \beta^*_1,\ldots,\beta_K^* &= \argmin_{n_k, \beta_k}\ E\left[E_1\bigg[\ell \bigg(  \hat \theta(\beta_{1:K}, n_{1:K}), X^{(1)}, y^{(1)}\bigg)-\ell(\theta_1, X^{(1)},y^{(1)})\bigg] \right]\\
    &\text{subject to} \quad \sum_{k=1}^K n_k = N, \sum_{k=1}^K \beta_k = 1 
    \end{split}
\end{align}

First for fixed $\beta$, we solve for $n_k$. In the following, for simplicity, we will treat the sample sizes $n_k$ as continuous.

Without loss of generality, let $\mathrm{Var}_f(\phi_1(D^{(f)}), \ldots, \phi_K(D^{(f)}))=I$. Denote $n \coloneqq (n_1, n_2, \ldots, n_K)$ and $\beta \coloneqq (\beta_1, \beta_2, \ldots, \beta_K)$. By Eqn \ref{eq:weak-convergence-excess-risk} (the Taylor expansion of the excess risk in Appendix \ref{app: thm1}), solving Eqn.~\ref{eqn: sample_constraint_obj} is asymptotically equivalent to solving %
\begin{align}\label{n_constraint}
\begin{split}
     n^*, \beta^* &= \argmin_{n, \beta }\ \mathbf{\beta}^\intercal(A\Sigma^W A^\intercal+ \text{Diag}(1/n_1, \ldots,1/n_K)) \mathbf{\beta}  \\
     &=\argmin_{n^\intercal \mathbf{1} = N, \mathbf{\beta}^\intercal \mathbf{1}=1}\ \mathbf{\beta}^\intercal (A\Sigma^W A^\intercal) \mathbf{\beta} + \sum_{k=1}^K \frac{\beta^2_k}{n_k},
    \end{split}
\end{align}
where $A$ and $\Sigma^W$ are as defined in Lemma \ref{lemma:samp_asymp}. Consider a relaxation of the sample size constraint and treat it as continuous. Let $ (n_1^*, n_2^*, \ldots, n_K^*)$ be an optimal solution to this problem. Then we can define a feasible solution by letting $\tilde{n}_{k'}=n_{k'}^*-\epsilon$, another coordinate $\tilde{n}_{k''}=n_{k''}^*+\epsilon$, and all other coordinates $\tilde{n}_{k}=n_k^*$ for some $\epsilon >0$. Then for a fixed $\beta$, 
\begin{align*}
    &\lim_{\epsilon \rightarrow 0} \ \frac{1}{\epsilon} \bigg(\frac{\beta_{k'}^2}{n^*_{k'}+\epsilon} + \frac{\beta_{k''}^2}{n^*_{k''}} - \frac{\beta_{k'}^2}{n^*_{k'}-\epsilon} - \frac{\beta_{k''}^2}{n^*_{k''} } \bigg) = 0\\
    &\Leftrightarrow |\beta_{k'}|=n_{k'}\frac{|\beta_{k''}|}{n_{k''}}
\end{align*} Since $\sum_{k=1}^K n_k = N$, summing over all $k'$ gives us that for any $k$, 
\begin{align}\label{opt_n_N}
    \frac{\sum_{k'} |\beta_{k'}|}{N} = \frac{|\beta_{k}|}{n_k},
\end{align}
giving us the optimal sample size
\begin{equation*}
    n_k^*=\frac{|\beta_{k}|}{\sum_{k'} |\beta_{k'}|}N.
\end{equation*} Substituting $n^*_k$ into the objective ~\ref{n_constraint} gives the simplified convex optimization problem.
\begin{align}
     \beta^* &= \argmin_{\mathbf{\beta}^\intercal \mathbf{1}=1}\ \beta^\intercal (A\Sigma^W A^\intercal) \beta + \frac{1}{N}\bigg(\sum_{k=1}^K |\beta_k|\bigg)^2.
\end{align}

By the above display, we can form a plug-in estimate of the optimal sampling sizes $n_k^*$ based on the optimal data set weights $\beta_k^*$.

In the following, let us consider the common special case that all optimal data set weights $\beta_k^*$ are positive. 
In that case, since $\sum_k \beta_k = 1$, we have
\begin{equation*}
    \frac{n_k^*}{N} = \beta_k^*.
\end{equation*}

Thus, in this case, the optimal \emph{sampling proportions }$n_k^*/N$ are equal to the optimal \emph{data set weights} $\beta_k^*$. Estimation of the optimal data set weights $\beta^*$ is discussed in  \cite{random_shift2}.

\subsection{Optimal sampling under budget constraint}\label{subsection: bud_constr}

Instead of a size constraint, we now consider a budget constraint where we can either sample from the target $P_1$ for a price of $\kappa_1 > 0$ per observation $(X_i^{(1)},y_{i}^{(1)})$ or sample from a cheaper but low quality data source $P_2$ at a cost of $\kappa_2 > 0$ per observation $(X_i^{(2)},y_{i}^{(2)})$. 
Recall that our weighted estimator is defined as 
\begin{equation*}
    \hat{\theta}(\beta_{1},\beta_2, n_{1},n_2) = \argmin_{\theta \in \Theta}   \beta_1 \frac{1}{n_1} \sum_{i=1}^{n_1} \ell(\theta,X_{i}^{(1)}, y_{i}^{(1)}) + \beta_2 \frac{1}{n_2} \sum_{i=1}^{n_2} \ell(\theta,X_{i}^{(2)}, y_{i}^{(2)}).
\end{equation*}
 
We are interested in the sample size and weights that minimize the excess risk:
\begin{align*}\label{opt_budget}
\begin{split}
    (n_1^*,n_2^*,\beta_1^*,\beta_2^*) &= \argmin_{n_1,n_2, \beta_1,\beta_2}\ E\left[E_1\left[\ell (\hat{\theta}(\beta_{1:2},n_{1:2}), X^{(1)},y^{(1)})-\ell(\theta_1, X^{(1)},y^{(1)}) \right] \right], \\
    &\text{subject to} \quad \kappa_1 n_1+\kappa_2 n_2 \leq C.
    \end{split}
\end{align*}
Here, we take an expectation of $(X^{(1)},y^{(1)})$ with respect to $P_1$, reflected by the inner expectation $E_1$. We then take expectation with respect to both sampling and distribution randomness of the estimator $\hat{\theta}$, reflected by the outer expectation $E$.

For fixed $n_1$ and $n_2$, we can solve for $\beta_1^*,\beta_2^*$. Using the asymptotic distribution of Lemma~\ref{lemma:samp_asymp}, we (approximately) have by Eqn \ref{eqn: optimal_beta} of Appendix \ref{app: thm1},

$$\beta^*_2 = \frac{\frac{1}{\Sigma^W_{2,2}+\Sigma^W_{1,1}-2\Sigma^W_{1,2}+1/n_2}}{n_1+\frac{1}{\Sigma^W_{2,2}+\Sigma^W_{1,1}-2\Sigma^W_{1,2}+1/n_2}},$$ 
leading to mean asymptotic excess risk (ignoring lower-order terms, see Eqn \ref{eq:weak-convergence-excess-risk} of Appendix \ref{app: thm1}) of
\begin{equation}\label{eq:excess-risk}
    \frac{1}{n_1 + \frac{1}{\Sigma^W_{2,2}+\Sigma^W_{1,1}-2\Sigma^W_{1,2}+1/n_2}} \cdot \text{Trace}(E_f(\nabla^2 \ell(\theta_f,X^{(f)},y^{(f)}))^{-1} \mathrm{Var}_f(\nabla \ell(\theta_f,X^{(f)},y^{(f)}))).
\end{equation}
Therefore, finding the optimal sample sizes $n_1^*, n_2^*$ can be done by solving the following problem: 

\begin{equation*}
   \arg \max_{n_1,n_2}  \left( n_1 + \frac{1}{\Sigma^W_{2,2}+\Sigma^W_{1,1}-2\Sigma^W_{1,2}+1/n_2}  \right) \text{ subject to } \kappa_1 n_1 + \kappa_2 n_2 \le C
\end{equation*}
Based on an estimate of $\Sigma^W$, which we will provide in Proposition \ref{prop: consistency_cov}, this can be solved using standard optimization algorithms such as gradient descent or Newton's method.

\subsection{Incremental value of data} 

In the following, we will discuss how adding one data point from each distribution will increase the target error. Consider the case when $W^{(1)}=1$, i.e. the target distribution is not perturbed from $P_f$. If we add a data point from distribution $P_1$, we increase the effective sample size in Equation~\eqref{eq:excess-risk} by approximately

\begin{equation*}
    \partial_{n_2} \left( \frac{1}{ \Sigma^W_{2,2} + \frac{1}{n_2}} + n_1  \right) = \frac{1}{\left(\Sigma^W_{2,2} \cdot n_2 + 1\right)^2 }
\end{equation*}
 Thus, we can either pay $ \$ \kappa_1 $ to increase the effective sample size by $1$ or  pay $ \$ \kappa_2$ to increase the effective sample size by approximately $ \frac{1}{(\Sigma^W_{2,2} \cdot n_2 + 1)^2}$. Note that $\frac{1}{(\Sigma^W_{2,2} \cdot n_2 + 1)^2} < 1$, so the additional data point from the low-quality source will always have less value than a data point from the high-quality source. Therefore, we would be willing to buy data from $P_2$, as long as the increase in effective sample size per dollar is larger for $P_2$ than for $P_1$, i.e.,

\begin{equation*}
    \frac{1}{\kappa_1} < \frac{1}{\kappa_2} \frac{1}{\left(\Sigma^W_{2,2} \cdot n_2 + 1\right)^2 },
\end{equation*}
which is equivalent to

\begin{equation*}
    \underbrace{\frac{\kappa_2}{\kappa_1}}_{\text{ratio of price per data point}}  \left( \underbrace{\Sigma^W_{2,2}}_{\text{data quality}} \cdot \underbrace{n_2}_{\text{data quantity}} + 1 \right)^2 < 1.
\end{equation*}
Based on an estimate of $\Sigma_{2,2}^W$, one can use this equation to estimate whether sampling the next data point from $P_1$ or $P_2$ is expected to lead to a larger drop in out-of-distribution error. We will provide a consistent estimator of $\Sigma_{2,2}^W$ in Proposition \ref{prop: consistency_cov}.

This equation has several implications on the optimal sampling strategy. First, we would not buy an additional observation from the low-quality source $P_2$ unless $\frac{\kappa_2}{\kappa_1} < 1$, i.e., the low-quality data source is cheaper than the high-quality data source. If $\Sigma^W_{2,2} > 0$, there is an interplay of data quality and data quantity. If the distribution shift is small ($\Sigma^W_{2,2} \approx 0$), then we tend to prefer sampling from the low-quality data source $P_2$. As $n_2 \rightarrow \infty$, there is a point where sampling from the low-quality data source $P_2$ increases the effective sample size less per dollar spent than sampling from the target data source $P_1$.

\subsection{Application: predicting U.S. income}\label{sec:empirical-size-constraint}
We now present an example to demonstrate our method for optimal sampling under size constraints described in Appendix \ref{subsection: bud_constr}. Suppose our goal is to predict the (log) income of an adult in the U.S., residing in any of the $50$ states or Puerto Rico. However, we can only collect $N$ total data samples from CA, PA, NY, FL, or TX due to resource constraints. How should we sample given this constraint? Intuitively, one option is to collect all samples from a single state, focusing on the state we heuristically deem most representative of the U.S. Alternatively, we could sample equally from each state to create a diverse training set. Figure~\ref{fig: us_opt_samp} compares the MSE on $20,000$ test observations for these strategies and our proposed optimal sampling method. While equal sampling across all states tends to improve performance relative to sampling from just one state, since the target is the entire U.S., the optimal sampling method drastically outperforms both. 

\begin{figure}[ht]
   \centering
   \includegraphics[width=0.55\linewidth]{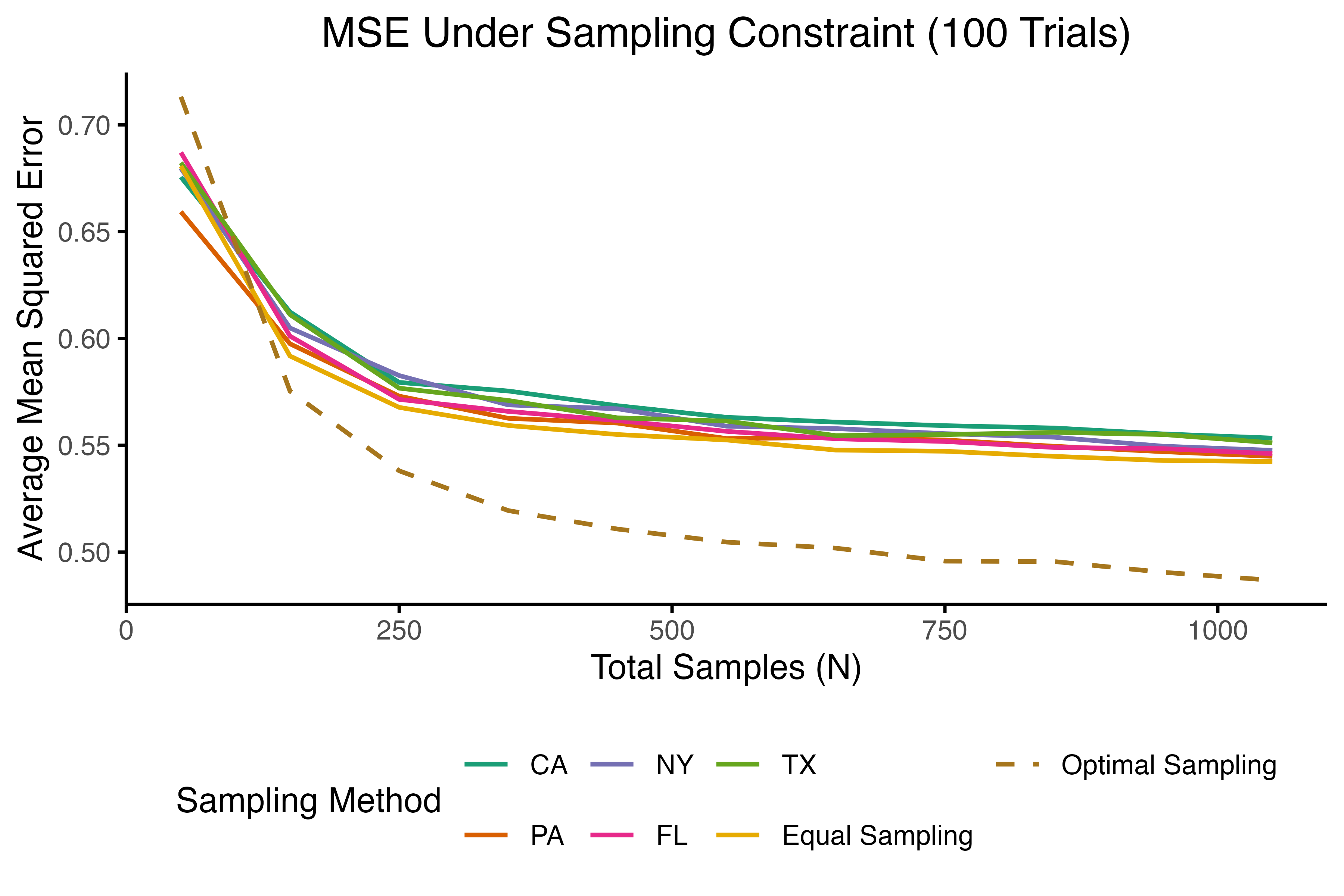}
   \caption{Results from prediction of an individual's $2018$ income in the U.S. using data from different states. Sampling from individual states for training does not generalize well when the test distribution is the US. While equal sampling of the five states is an improvement, our optimal sampling method outperforms the other strategies. Our procedure does not use outcome data for determining the optimal sampling ratios.}
   \label{fig: us_opt_samp}
\end{figure}

\section{Additional empirical results for data prioritization} 
\subsection{Importance of subsampling for DUC estimation}\label{app: acs}
We first show an example of DUC estimation from 1 trial (no subsampling) in Figure \ref{fig:ACS_onetrial} below.  
\begin{figure}[H]
    \centering
    \includegraphics[width=0.85\linewidth]{ACS_one_trial.png}
    \caption{Randomly sampled trial illustrating data prioritization method from 10 state datasets.  }
    \label{fig:ACS_onetrial}
\end{figure}
We now use subsampling and average over 100 trials, resulting in a correlation of $-0.94$ between DUC and average ranking. 
\begin{figure}[H]
    \centering
    \includegraphics[width=0.85\linewidth]{acs_100.png}
    \caption{Average ranking, from lowest to highest test MSE, using Random Forest trained on California ($n_\text{CA} =30$) plus one source state ($n_\text{source} =1000$) over 100 trials. Compared to more subsampling as in Fig \ref{fig: pred_CA}, DUC estimation exhibits weaker correlation with ranking.}
    \label{fig:acs_100}
\end{figure}
Increasing to 1,000 trials, we get Figure \ref{fig: pred_CA} from Section \ref{sec: acs_income}, giving us a stronger correlation of $-0.96$ between DUC and average ranking, illustrating how subsampling and averaging estimates over trials can improve DUC estimation. 

We now demonstrate a second scenario where subsampling can improve estimation: when you have a poor estimate of the target population covariate means. We mimic this scenario by estimating $E_1[X]$ using only $50$ samples in the below experiment. In general, higher number of trials improve DUC estimation compared to using only 100 trials.    
\begin{table}[ht!]
\centering
\caption{Correlation between the average estimated DUC and the average rank of test MSE across different trial counts ($n_{\text{sample}} = 50$). }
\label{tab:duc_corr_trials}
\begin{tabular}{@{}lccccc@{}}
\toprule
\textbf{Number of Trials} & \textbf{100} & \textbf{200} & \textbf{300} & \textbf{400} & \textbf{500} \\ 
\midrule
Correlation & $-0.919$ & $-0.936$ & $-0.947$ & $-0.949$ & $-0.938$ \\
\bottomrule
\end{tabular}
\end{table}

\subsection{NHANES}\label{app: nhanes}

\begin{table}[H]
\centering
\begin{threeparttable}
\caption{Composition of candidate source datasets for 2017 cholesterol prediction}
\label{tab:nhanes_sources}
\begin{tabular}{llll}
\hline
\textbf{Dataset} & \textbf{\texttt{nhanes\_year}} & \textbf{\texttt{RIAGENDR}} & \textbf{\texttt{RIDRETH\_merged}} \\
\hline
 1 & $2013, 2015$ & 0      & 1 \\
 2 & $1999, 2001, 2003, 2015$ & All & All \\
 3 & $2013, 2015$ & 1      & 1 \\
 4 & $2011$       & 1      & All \\
 5 & \texttt{RIAGENDR} = 0: $1999, 2015$; \texttt{RIAGENDR} = 1: $2000, 2015$ & 0 / 1 & 0 \\
 6 & $1999, 2005, 2013$ & All & All \\
\hline
\end{tabular}
\begin{tablenotes}[flushleft]
\footnotesize
\item \texttt{nhanes\_year}: survey year.
\item \texttt{RIAGENDR}: gender of the participant.
\item \textbf{\texttt{RIDRETH\_merged}}: derived feature using \texttt{RIDRETH3} (Recode of reported race and Hispanic origin information, with Non-Hispanic Asian Category) where available, otherwise \texttt{RIDRETH1} (Recode of reported race and Hispanic origin information). See \texttt{tableshift} \citep{gardner2023tableshift} for recoding details.
\end{tablenotes}
\end{threeparttable}
\end{table}

\begin{figure}[H]
    \centering
    \includegraphics[width=0.85\linewidth]{fig_nhanes_mse.png}
    \caption{Average test MSE of optimally weighted samples from 2017 ($n_{\text{2017}}=30$) and a source dataset ($n_{\text{source}}=1000$) over 1000 trials (baseline average test MSE of $1.20$ with only $n_{\text{2017}}=30$ target samples). Higher $\hat{\rho}^2$ (DUC) correlates more strongly with lower test MSE compared to the KL divergence and domain classifier score.}
\end{figure}

\end{document}